\documentclass[11pt,twoside,letterpaper]{article}

    \usepackage[english]{babel}
    \usepackage[T1]{fontenc}
    \usepackage[utf8x]{inputenc}
    \usepackage{amsmath,amsfonts,amssymb,amsthm}
    \usepackage{xspace}
    \usepackage{fancyhdr}
    \usepackage[langfont=typewriter]{complexity}
    \usepackage{geometry}
    \usepackage{lscape}
    
    \usepackage{subfig}
    \theoremstyle{plain}
    \newtheorem{theorem}{Theorem}
    \newtheorem{lemma}{Lemma}
    \newtheorem{corollary}{Corollary}
    \newtheorem{proposition}{Proposition}
    \theoremstyle{definition}
    \newtheorem{definition}{Definition}

\usepackage{color,graphicx}

\usepackage{mynatbib}
\newcommand{\EKbibtie}{\ }
\def\citet{\citep} 
\newcommand{\citeyearp}[1]{[\citeyear{#1}]}

\bibpunct[, ]{[}{]}{;}{a}{}{,}

\usepackage{enumerate}

\newcommand{\EKhref}[2]{
    URL: \href{#1}{\nolinkurl{#2}}
    }
\usepackage[
breaklinks
,colorlinks
,linkcolor=red
,anchorcolor=blue
,pagecolor=blue
,citecolor=blue
,bookmarks=true
]{hyperref}

\usepackage{breakurl}

\newcommand{\RR}{{\mathbb R}}
\newcommand{\CC}{{\mathbb C}}
\newcommand{\FF}{{\mathbb F}}

\DeclareMathOperator\gsize{gsize}
\DeclareMathOperator\PER{PER}
\DeclareMathOperator\DETfam{DET}
\DeclareMathOperator\per{per}
\DeclareMathOperator\domdef{def}
\DeclareMathOperator\BoolP{BP}
\DeclareMathOperator\Mon{Mon}

\chardef\ttlb="7B 
\chardef\ttrb="7D 
\chardef\ttti="7E

\begin{document}

\title{
Symmetric Determinantal Representation of Formulas\\
and Weakly Skew Circuits
}

\author{
Bruno Grenet$^{\textrm{*}}$
\and Erich L. Kaltofen$^{\textrm{\dag}}$
\and Pascal Koiran$^{\textrm{*}}$
\and Natacha Portier$^{\textrm{*\ddag}}$
}

\date{}

\maketitle

\vspace{3em}

\begin{center}
\framebox[1.1\width]{Rapport de Recherche RRLIP2010-24}
\end{center}

\vspace{2em}

\footnotetext{
$^{\textrm{*}}$LIP, UMR 5668, ENS de Lyon -- cnrs -- UCBL -- INRIA,
\'Ecole Normale Sup\'erieure de Lyon, Université de Lyon
and Department of Computer Science, University of Toronto \\
{\ttfamily \{Bruno.Grenet,Pascal.Koiran,Natacha.Portier\}@ens-lyon.fr}}

\footnotetext{
$^{\textrm{\dag}}$ 
Dept.\ of Mathematics, North Carolina State University,
Raleigh, North Carolina 27695-8205, USA\\
{\ttfamily kaltofen@math.ncsu.edu};
\url{http://www.kaltofen.us}\\
This material is based on work supported in part
by the National Science Foundation under Grants
CCF-0830347 and CCF-0514585.}

\footnotetext{
$^{\textrm{\ddag}}$ partially funded by European Community under contract PIOF-GA-2009-236197 of the 7th PCRD.}

\begin{abstract}
We deploy algebraic complexity theoretic techniques to construct
symmetric determinantal representations of formulas and weakly skew
circuits.  Our representations  produce
matrices of much smaller dimensions than those given in the convex
geometry literature when applied to polynomials having a concise
representation (as a sum of monomials, or more generally as an arithmetic 
formula or a weakly skew circuit).
These representations are valid in any field of characteristic different 
from 2. In characteristic 2 we are led to an almost complete solution
to a question of  Bürgisser on the $\VNP$-completeness 
of the partial permanent.
In particular, we show that the  partial permanent cannot be $\VNP$-complete
in a finite field of characteristic 2 unless the polynomial hierarchy collapses.
\end{abstract}

\section{Introduction} 

\subsection{Motivation} 
\label{sec:motiv}

A linear matrix expression
(symmetric linear matrix form, 
affine symmetric matrix pencil) 
is a symmetric matrix with the
entries being linear forms in the variables $x_1,\ldots,x_n$ and
real number coefficients:
\begin{equation}\label{eq:LME}
A(x_1,\ldots,x_n) = A_0 + x_1 A_1 + \cdots + x_n A_n,
\text{\quad $A_i$ symmetric in $\RR^{t\times t}$}.
\end{equation}
A linear matrix inequality (LMI) restricts to those values $\xi_i\in \RR$ of the $x_i$
such that $A(\xi_1,\ldots,\xi_n)\succeq 0$, i.e., is positive semidefinite.
The set of all such values defines a spectrahedron.

A \emph{real zero polynomial} is a polynomial $p$ with real coefficients such that for every $x\in\RR^n$ and every $\mu\in\CC$, $p(\mu x)=0$ implies $\mu\in\RR$. The Lax conjecture and generalized Lax conjecture seek for representations of real zero polynomials $f(x_1$, $\ldots$, $x_n)$ 
as $f = \det(A)$ with $A$ as in  \eqref{eq:LME} and $A_0 \succeq 0$. This is in fact an equivalent formulation of the original Lax conjecture which was stated in terms of hyperbolic polynomials (see~\citep{LPR05} for this equivalence). Furthermore, the matrices are required to have dimension $d$ where $d$ is the degree of the polynomial. For $n=2$ such representations always exist while a counting argument shows that this is impossible for $n>2$ \citep{HV06} (actually, 
\citep{LPR05} give the first proof of the Lax conjecture in its original form based on the results of 
\citep{HV06}). Two relaxations have been suggested to 
evade
this counting argument: At first it was suggested to remove the dimension constraint and seek for bigger matrices, and this was further relaxed by seeking for representations of some power of the input polynomial. Counterexamples to both relaxations have recently been constructed \citep{Branden2010}. 

Another relaxation is to drop the condition $A_0 \succeq 0$ and represent any $f$ as $\det(A)$ \citep{HMcCV06,Quarez08}.  However, the purely algebraic construction of \citep{Quarez08} leads to exponential matrix dimensions $t$. Here we continue the line of work initiated in \citep{HMcCV06,Quarez08} but we proceed differently by symmetrizing the complexity theoretic construction by Valiant \citeyearp{Val79}.  Our construction yields smaller dimensional matrices not only for polynomials represented as sums of monomials but also for polynomials represented by formulas and weakly skew circuits \citep{MP08,KaKoi08}. Even though in the most general case the bounds we obtained are slightly 
worse than Quarez's \citeyearp{Quarez08}, in a lot of interesting cases such as polynomials with a polynomial size formula or weakly-skew circuit, or in the case of the permanent, our constructions yield much smaller matrices (see Section~\ref{sec:comp}). 

Our constructions are valid for any field of characteristic different from $2$.
For fields of characteristic $2$, 
it can be shown that some polynomials (such as e.g. the polynomial $xy+z$)
cannot be represented as determinants of symmetric matrices~\citep{GMT}.
Note as a result that the 2-dimensional permanent $xw+yz$ 
cannot be ``symmetrized''
over characteristic~$2$ with any dimension.
It would be interesting to exactly characterize which polynomials admit such a
representation in characteristic $2$.
For the polynomial $x+y$, we have 
$$
  x+y = \det(
  \begin{bmatrix}
  0 & x & 0 & y & -1\\
  x & 0 & 1 & 0 & 0 \\
  0 & 1 & 0 & -1 & 0 \\
  y & 0 & -1 & 0 & 1/2\\
  -1 & 0 & 0 & 1/2 & 0
  \end{bmatrix}
  )
  = \det(
  \begin{bmatrix}
  x & 0 & 0 & 1\\
  0 & y & 0 & 1\\
  0 & 0 & 1 & 0\\
  1 & 1 & 0 & 0
  \end{bmatrix}
  ),
$$
where the first matrix is derived from our construction, but the second
is valid over any commutative ring.
It is easily shown 
that for every
polynomial $p$, its square $p^2$ admits a symmetric determinantal representation in
characteristic $2$.
This is related to a question of B\"urgisser
\citeyearp{Burgisser}: Is the partial permanent $\VNP$-complete over fields of
characteristic $2$? We give an almost complete negative answer to this
question.

Our results give as a by-product an interesting result which was not known to the authors' knowledge: Let $A$ be an $(n\times n)$ matrix with indeterminate coefficients (ranging over a field of characteristic different from $2$), then there exists a symmetric matrix $B$ of dimensions 
$O(n^5)$ which entries are the indeterminates from $A$ and constants from the field such that $\det A=\det B$. This relies on the existence of a size-$O(n^5)$ weakly-skew circuit to compute the determinant of an $(n\times n)$ matrix \citep{Ber84,MP08}, and this weakly-skew circuit can be represented by a determinant of a symmetric matrix as proved in this paper. 
The dimensions 
of $B$ can be reduced to $O(n^4)$ if we replace the weakly skew
circuits from \citep{Ber84,MP08} by the  skew circuits of size $O(n^4)$
constructed by Mahajan and Vinay~\citeyearp{MV97}. These authors construct an arithmetic
branching program for the determinant with $O(n^4)$ edges,\footnote{This bound can be found on p.11 of their paper.}
and the arithmetic branching program can
be evaluated by a skew
circuit of size $O(n^4)$.
After learning of our result, Meena Mahajan 
and Prajakta Nimbhorkar have noticed that the 
arithmetic branching program for the determinant
can be transformed directly into a symmetric determinant of dimensions 
$O(n^3)$
with techniques similar to the ones used 
in this paper. 
We give a detailed proof in Subsection~\ref{sec:Meena}.

We add that the assymptotically smallest known division-free algebraic circuits for the $n\times n$ determinant
polynomial have size $O(n^{2.70})$ \citep{Ka92:issac,KaVi04:2697263}.  The circuits actually can
compute the characteristic polynomial and the adjoint and are based on algebraic rather than combinatorial
techniques.  Weakly skew circuits of such size appear not to be known.

\vspace{1ex}
\noindent{\bfseries Organization.} 
Section~\ref{sec:known} is devoted to an introduction to the algebraic complexity theoretic used in our constructions, as well as a reminder of the existing related constructions in algebraic complexity. Section~\ref{sec:formulas} deals with symmetric representations of formulas while Section~\ref{sec:ws-circuits} focuses on weakly-skew circuits. Table~\ref{tab:summary} page~\pageref{tab:summary} gives an overview of all the different constructions used in this paper. Section~\ref{sec:comp} then proceeds to the comparisons between the results obtained so far and Quarez's~\citeyearp{Quarez08}. The special case of fields of characteristic $2$ is studied in Section~\ref{sec:char2}.\\
A shorter version of this paper \citep{GKKP11} has been published in Proceedings of STACS 2011. 
It contains material from Section~\ref{sec:ws-circuits} and Section~\ref{sec:char2}.

\vspace{1ex}
\noindent{\bfseries Acknowledgments.}
We learned of the symmetric representation problem from Markus Schweighofer's
ISSAC 2009 Tutorial\\ 
\url{http://www.math.uni-konstanz.de/\~{}schweigh/presentations/dcssblmi.pdf}.

We thank Meena Mahajan for pointing out~\citep{MV97}, sketching 
the 
construction of a symmetric determinant of dimensions 
$O(n^3)$ from a
determinant of dimensions 
$n$ 
and reading our proof of it.

\subsection{Known results and definitions} 
\label{sec:known}

In his seminal paper Valiant \citeyearp{Val79} expressed the polynomial computed by an arithmetic formula as the determinant of a matrix whose entries are constants or variables. If we define the \emph{skinny size} $e$ of the formula as its number of arithmetic operations then the dimensions of the matrix are 
at most $e+2$. The proof uses a weighted digraph construction where the formula is encoded into paths from a source vertex to a target, sometimes known as an Algebraic or Arithmetic Branching Program \citep{N91,BG99}.
This theorem shows that every polynomial with a sub-exponential size formula can be expressed as a determinant with sub-exponential dimensions, 
enhancing the prominence of linear algebra. A slight variation of the theorem is also used to prove the universality of the permanent for formulas which is one of the steps in the proof of its $\VNP$-completeness. In a tutorial, von zur Gathen \citeyearp{vzG87} gives another way to express a formula as a determinant: his proof does not use digraphs and his bound is $2e+2$. Refining his techniques, Liu and Regan \citeyearp{LR06} gave a construction leading to an upper bound of $e+1$ in 
a slightly more powerful model: 
multiplications by constant are free and do not count into the size of the formula. 

Our purpose here is to express a formula as a determinant of a symmetric matrix. Multiplications by constant are also given for free. Our construction uses paths in graphs, similar to the paths in digraphs in Valiant's original proof. In fact, this original construction appears to have a little flaw in it. Interestingly enough, this flaw has never been mentioned in the literature to the authors' knowledge. A slight change in the proof is given in \citep[Exercise 21.7 (p570)]{BurgisserClausenShokrollahi} that settles a part of the problem. And the same flaw appears in the proof of the universality of the permanent in \citep{Burgisser}. When adding two formulas, the resulting digraph can have two arcs between the source and the target, which can lead to the sum of two variables being an entry of the matrix, and this is not allowed as we seek for symmetric matrices where each entry is either a constant or a variable. 
The first idea to correct the proof is to keep the same parity for all $s$-$t$-paths 
as in Valiant's original proof, adding two new vertices and replacing one of the arcs by a length-three path. This method is very simple but its disadvantage is that it increases the dimensions 
of the final matrix to $2e+3$. In the symmetric case we will use a 
$-1$ coefficient to correct the parity differences between paths instead of adding new vertices. Using this technique in the non-symmetric case allows us to prove Valiant's theorem with $(e+1)$ instead of $(e+2)$. Our technique also gives for free multiplications by constants as in \citep{LR06}. It uses digraphs and is to our opinion more intuitive than direct work on matrices.

In \citep{Tod92,MP08}, results of the same flavor were proved for a more general class of circuits, namely the \emph{weakly-skew} circuits. Malod and Portier~\citeyearp{MP08} can deduce from those results a fairly simple proof of the $\VQP$-completeness of the determinant (under $qp$-projection). Moreover, they define a new class $\VP_{ws}$ of polynomials represented by polynomial-size weakly-skew circuits (with no explicit restriction on the degree of the polynomials) for which the determinant is complete under $p$-projection. A formula is a circuit in which every vertex has out-degree $1$ (but the output). This means in particular that the underlying digraph is a tree. A weakly-skew circuit is a kind of generalization of a formula, with a less constrained 
structure on the underlying digraph. For an arithmetic circuit, the only restriction on the digraph is the absence of directed cycles (that is the underlying digraph is a directed acyclic graph). A circuit is said weakly-skew if every multiplication gate $\alpha$ has the following property: the sub-circuit associated with one of its arguments $\beta$ is connected to the rest of the circuit only by the arrow going from $\beta$ to $\alpha$. This means that the underlying digraph is disconnected as soon as the multiplication gate $\alpha$ is removed. In a sense, one of the arguments of the multiplication gate was specifically computed for this gate.

Toda \citeyearp{Tod92} proved that the polynomial computed by a weakly-skew circuit of skinny size $e$ can be represented by the determinant of a matrix of dimensions 
$(2e+2)$. This result was improved by Malod and Portier \citeyearp{MP08}: The construction leads to a matrix of dimensions 
$(m+1)$ where $m$ is the \emph{fat size} of the circuit (\emph{i.e.} its total number of gates, including the input gates). Note that for a circuit in general and for a weakly-skew circuit in particular $m\le 2e+1$. The latter construction uses negated variables in the matrix. It is actually possible to get rid of them \citep{KaKoi08}. Although the skinny size is well suited for the formulas, the fat size appears more appropriate for weakly-skew circuits. In Section~\ref{sec:ws-circuits}, we symmetrize this construction so that a polynomial expressed by a weakly-skew circuit equals the determinant of a symmetric matrix. Our construction yields a symmetric matrix of dimensions $(2m+1)$. 
In fact, this can be refined as well as the non-symmetric construction. An even more appropriate size for a weakly-skew circuit is $(e+i)$ where $e$ is the skinny size and $i$ the number of inputs labelled by a variable (clearly $e+i\le m$). We can show that the bounds are still valid if we replace $m$ by $(e+i)$ and even when multiplications by constants are free as in \citep{LR06} (see Section~\ref{sec:green-size}). 

Let us now give some formal definitions of the arithmetic circuits and related notions.

\begin{definition}
An \emph{arithmetic circuit} is a directed acyclic graph with vertices of in-degree $0$ or $2$ and exactly one vertex of out-degree $0$. Vertices of in-degree $0$ are called \emph{inputs} and labelled by a constant or a variable. The other vertices, of in-degree $2$, are labeled by $\times$ or $+$ and called \emph{computation gates}. 
The vertex of out-degree $0$ is called the \emph{output}. The vertices of a circuit are commonly called \emph{gates} and its arcs \emph{arrows}. 
\end{definition}

An arithmetic circuit with constant inputs in a field $k$ and variables in a set $\bar x$ naturally computes a polynomial $f\in k[\bar x]$. 

\begin{definition}
If $\alpha$ is a gate of a circuit $C$, the \emph{sub-circuit associated to $\alpha$} is the subgraph of $C$ made of all the gates $\beta$ such that there exists a oriented path from $\beta$ to $\alpha$ in $C$, including $\alpha$. A gate $\alpha$ receiving arrows from $\beta$ and $\gamma$ is said to be \emph{disjoint} if the sub-circuits associated to $\beta$ and $\gamma$ are disjoint from one another. The gates $\beta$ and $\gamma$ are called the \emph{arguments} of $\alpha$.
\end{definition}

\begin{definition}
An arithmetic circuit is said \emph{weakly-skew} if for any multiplication gate $\alpha$, the sub-circuit associated to one of its arguments $\beta$ is only connected to the rest of the circuit by the arrow going from $\beta$ to $\alpha$: it is called the \emph{closed} sub-circuit of $\alpha$. 
A gate which does not belong to a closed sub-circuit of $C$ is said to be \emph{reusable} in $C$. 

A \emph{formula} is an arithmetic circuit in which all the gates are disjoint.
\end{definition}

The reusability of a gate depends of course on the considered circuit $C$. For instance, in Fig.~\ref{fig:circuit-weak-form}\subref{fig:ex-weak}, the weakly-skew circuit 
has two closed sub-circuits. The input $z$ is in the right closed sub-circuit and is therefore not reusable. But inside this closed sub-circuit, it is reusable, and actually used as argument to the summation gate twice. Figures~\ref{fig:circuit-weak-form}\subref{fig:ex-circuit} and \subref{fig:ex-form} 
are respectively an equivalent arithmetic circuit and an equivalent formula, that is the two circuits and the formula compute the polynomial $(x+y)^2+2yz$. 

Let us remark a fact that will be useful later: all the multiplication gates of a weakly-skew circuit are disjoint  (but this is not a sufficient condition).
\begin{figure}[tbp]
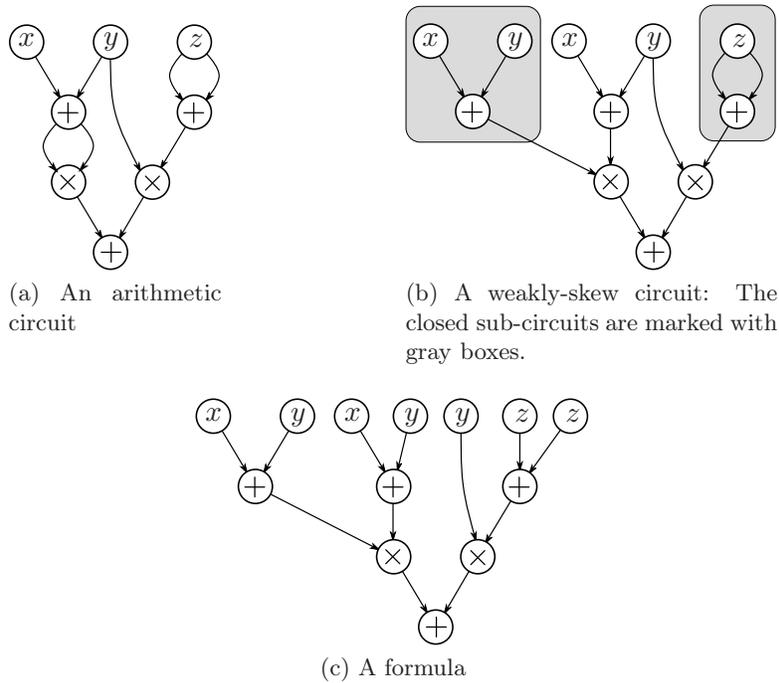

\centering
\subfloat[An arithmetic circuit]{\label{fig:ex-circuit}\input{ex-circuit.svg_tex}}\hfil
\subfloat[A weakly-skew circuit: The closed sub-circuits are marked with gray boxes.]{\label{fig:ex-weak}\input{ex-weak.svg_tex}}\\
\subfloat[A formula]{\label{fig:ex-form}\input{ex-form.svg_tex}}
\caption{An arithmetic circuit, a weakly-skew circuit and a formula computing the same polynomial $(x+y)^2+2yz$.}
\label{fig:circuit-weak-form}
\end{figure}

In our constructions, we shall use \emph{graphs} and \emph{digraphs}. In particular, the improved construction based on Valiant's represents formulas by paths in a digraph. On the other hand, to obtain symmetric determinantal representations the digraphs have to be symmetric. These correspond to graphs. 
In order to avoid any confusion between directed and undirected graphs, we shall exclusively use the term graph for undirected ones, and otherwise use the term digraph. It is well-known that cycle covers in digraphs are in one-to-one correspondence with permutations of the vertices and therefore that the permanent of the adjacency matrix of a digraph can be defined in terms of cycle covers of the digraph. 
Let us now give some definitions for those facts, and see how it can be extended to graphs.

\begin{definition}
A \emph{cycle cover} of a digraph $G=(V,A)$ is a set of cycles such that each vertex appears in exactly one cycle. The \emph{weight} of a cycle cover is defined to be the product of the weights of the arcs used in the cover. Let the \emph{sign} of a vertex cover be the sign of the corresponding permutation of the vertices, that is $(-1)^N$ where $N$ is the number of even cycles. Finally, let the \emph{signed weight} of a cycle cover be the product of its weight and sign.

For a graph $G=(V,E)$, let $G^d=(V,A)$ be the corresponding symmetric digraph. Then a cycle cover of $G$ is a cycle cover of $G^d$, and the definitions of weight and sign are extended to this case. In particular, if there is a cycle cover of $G$ with a cycle $C=(u_1,\dots,u_k)$, then a new cycle cover is defined if $C$ is replaced by the cycle $(u_k,\dots,u_1)$. Those two cycle covers are considered as different cycle covers of $G$.
\end{definition}

\begin{figure}[tbp]
\centering
\includegraphics[scale=1.2]{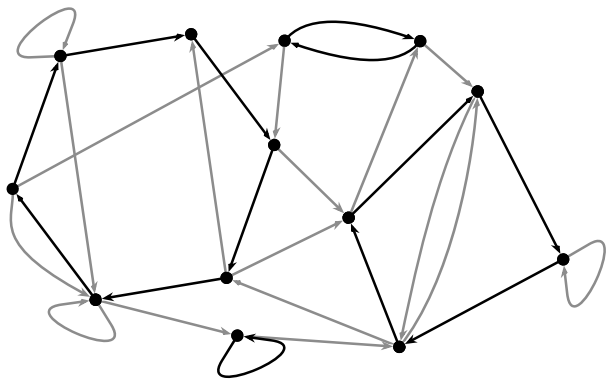}
\caption{A graph with a cycle cover (the arcs in the cover are in black).}
\label{fig:cycle-cover}
\end{figure}

\begin{definition}
Let $G$ be a digraph. Its \emph{adjacency matrix} is the $(n\times n)$ matrix $A$ such that $A_{i,j}$ is equal to the weight of the arc from $i$ to $j$ ($A_{i,j}=0$ is there is no such arc). The definition is extended to the case of graphs, seen as symmetric digraphs. In particular, the adjacency matrix of a graph is symmetric.
\end{definition}

\begin{lemma}\label{lemma:perm-det}
Let $G$ be a (di)graph, and $A$ its adjacency matrix. Then the permanent of $A$ equals the sum of the weights of all the cycle covers of $G$, and the determinant of $A$ is equal to the sum of the signed weights of all the cycle covers of $G$.
\end{lemma}

\begin{proof}
The cycle covers are obviously in one-to-one correspondence with the permutations of the set of vertices, and the sign of a cycle cover is defined to match the sign of the corresponding permutation. Suppose that the vertices of $V$ are $\{1,\dots,n\}$ and let $A_{i,j}$ be the weight of the arc $(i,j)$ in $G$. Let $C$ a cycle cover and $\sigma$ the corresponding permutation. Then it is clear that the weight of $C$ is $A_{1,\sigma(1)}\cdots A_{n,\sigma(n)}$, hence the result.
\end{proof}

The validity of this proof for graphs follows from the definition of the cycle covers of a graph in terms of the cycle covers of the corresponding symmetric digraph. In the following, the notion of perfect matching is used. A \emph{perfect matching} in a graph $G$ is a set $M$ of edges of $G$ such that every vertex is incident to exactly one edge of $M$. The weight of a perfect matching is defined in this as the weight of the corresponding cycle cover (with length-$2$ cycles). This means that this is the product of the weights of the arcs it uses, or equivalently it is the square of the product of the weights of the edges it uses. Note that this is the square of the usual definition.

A \emph{path} $P$ in a digraph is a subset of vertices $\{u_1,\dots,u_k\}$ such that for $1\le i\le k-1$, there exists an arc from $u_i$ to $u_{i+1}$ with nonzero weight. The size $|P|$ of such a path is $k$.

\section{Formulas} 
\label{sec:formulas}

\subsection{Non-symmetric case} 
\label{sec:valiant}

In this section, as in Sections~\ref{sec:sym-form} and \ref{sec:ws-circuits}, a field $k$ of characteristic different from $2$ is fixed and the constant inputs of the formulas and the weakly-skew circuits are taken from $k$. The variables are supposed to belong to a countable set $\bar x=\{x_1,x_2,\dots\}$. Following \citep{LR06}, we define a formula size that does not take into account multiplications by constants.

\begin{definition}\label{def:green-size}
Consider formulas with inputs being variables or constants from $k$. The green size $\gsize(\varphi)$ of a formula $\varphi$ is defined inductively as follows:
\begin{itemize}
\item The green size of a constant or a variable is $0$;
\item If $c$ is a constant then the green size of $c \times \varphi$ is equal to the green size of $\varphi$;
\item If $\varphi_1$ and $\varphi_2$ are formulas, then $\gsize(\varphi_1+\varphi_2)=\gsize(\varphi_1)+\gsize(\varphi_2)+1$.
\item If $\varphi_1$ and $\varphi_2$ are non-constant formulas, then $\gsize(\varphi_1\times\varphi_2)=\gsize(\varphi_1)+\gsize(\varphi_2)+1$
\end{itemize}
\end{definition}

An even smaller size can be defined by deciding that every variable-free formula has size zero and Theorem~\ref{liu-regan} can easily be extended to this case. 
A formal definition of this size is given is Section \ref{sec:green-size} in the context of weakly-skew circuits.

\begin{theorem}[\citep{LR06}]
\label{liu-regan}
For every formula $\varphi$ of green size $e$ with at least one addition there is a square matrix $A$ of dimensions 
$e+1$ whose entries are inputs of the formula and elements of $\{0,1,-1,1/2\}$ such that $\varphi=\det(A)$.
\end{theorem}

We remark 
that if $\varphi$ has no addition it is of the form $cx_1\dots x_n$ and it has size $(n-1)$. Then a suitable matrix is the $(n+1)\times(n+1)$ diagonal matrix made of the $n$ variables and the constant $c$. Thus the dimensions of the matrix are at most 
$n+1=e+2$, and are $n=e+1$ if $c=1$. Note that this latter bound is minimal as the determinant of a $(d\times d)$ matrix is a degree-$d$ polynomial. The dimensions $(n+1)$ are 
not minimal when $c\neq 1$ as shown by the $(3\times 3)$ matrix 
\[\begin{bmatrix} 0&x&y\\ x&0&z\\ y&z&0\end{bmatrix}\]
representing $2xyz$. One can also see that the $n$ bound cannot be general as there is no $(2\times 2)$ matrix representing the polynomial $2xy$.

\begin{lemma}
\label{construction-pas-sym}
Let $\varphi$ be an arithmetic formula of green size $e$. Then there exists a constant $c_0$ and an edge-weighted digraph $G$ with at most $e+2$ vertices and two distinct vertices $s$ and $t$ such that
\[c_0\cdot\sum_{\text{$s$-$t$-path $P$}}(-1)^{|P|} \ w(P) = \varphi.\]
\end{lemma}

\begin{proof}[Proof of Lemma~\protect\ref{construction-pas-sym}]
We prove the lemma by induction on formulas. If $\varphi$ is equal to a variable $x$ (resp. a constant $c$) then $G$ has two vertices $s$ and $t$ and an edge $(s,t)$ labelled by $x$ (resp. $c$) and the constant $c_0$ is equal to $1$. 

If $\varphi=c \times \varphi'$ let $G'$ be the digraph and $c'_0$ the constant satisfying the lemma for the formula $\varphi'$. Then obviously $G=G'$ and $c_0=c'_0 c$ satisfy the lemma for $\varphi$. 

If $\varphi=\varphi_1\times \varphi_2$, let $G_1$ and  $c_1$  (resp. $G_2$ and  $c_2$) satisfying the lemma for $\varphi_1$ (resp. $\varphi_2$). Then let $c=c_1c_2$ and $G$ be the disjoint union of $G_1$ and $G_2$, except for $t_1$ and $s_2$ which are merged (see Fig~\ref{fig:val-prod-cst}). 

    \begin{figure}[htbp]
    \begin{center}
    \input{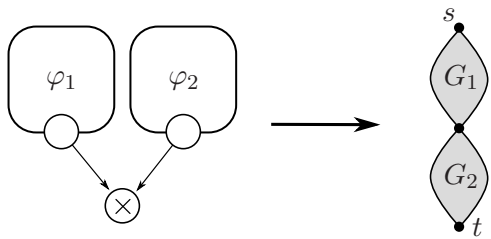}
    \caption{$G_1,c_1$ and $G_2,c_2$ are respectively associated to $\varphi_1$ and $\varphi_2$; $\varphi=\varphi_1\times \varphi_2$.}
    \label{fig:val-prod-cst}
    \end{center}
    \end{figure}
The size of $G$ is equal to $|G_1|+|G_2|-1 \leq \gsize(\varphi_1)+\gsize(\varphi_2)+3=\gsize(\varphi)+2$. A $s$-$t$-path $P$ in $G$ is a $s_1$-$t_1$-path $P_1$ in $G_1$ followed by a $s_2$-$t_2$-path $P_2$ in $G_2$ and we have $|P|=|P_1|+|P_2|-1$ and $w(P)=w(P_1)\times w(P_2)$, hence the result. 

If $\varphi=\varphi_1 + \varphi_2$, let $G_1$ and  $c_1$  (resp. $G_2$ and  $c_2$) satisfying the lemma for $\varphi_1$ (resp. $\varphi_2$). If $c_1=0$ then $\varphi$ and $\varphi_2$ compute the same polynomial and we just have to take $G=G_2$ and $c=c_2$. Suppose now $c_1 \neq 0$. Then we define $G$ as the disjoint union of $G_1$ and $G_2$, except for $s_1$ and $s_2$ which are merged, and with an edge $(t_2, t_1)$ of weight $-c_2/c_1$ (see Fig~\ref{fig:val-sum-cst}). 

    \begin{figure}[htbp]
    \begin{center}
    \input{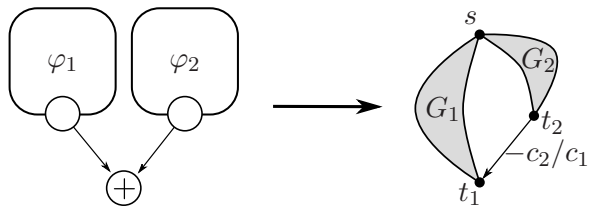}
    \caption{$G_1,c_1$ and $G_2,c_2$ are respectively associated to $\varphi_1$ and $\varphi_2$; $\varphi=\varphi_1+\varphi_2$.}
    \label{fig:val-sum-cst}
    \end{center}
    \end{figure}
The size of $G$ satisfies the same relation as in the multiplication case. Let $c_0=c_1$. A $s$-$t$-path $P$ in $G$ is a $s_1$-$t_1$-path in $G_1$ or a $s_2$-$t_2$-path $P_2$ in $G_2$ followed by the edge $(t_2, t_1)$, and in the second case we have $w(P)=w(P_2)(-c_2/c_1)$ and $|P|=|P_2|+1$, hence the result. Remark that $t_2$ has only one outgoing edge and its weight is a constant, and that this property will not be changed in the inductive construction. This property will be useful to prove the bound in the theorem.
\end{proof}

\begin{proof}[Proof of Theorem~\protect\ref{liu-regan}]
Let $\varphi$ be an arithmetic formula of green size $e$ and let $G$ and $c_0$ be given by Lemma~\ref{construction-pas-sym}. Let $\bar{G}$ be the digraph obtained from $G$ in the following way. We merge $s$ and $t$. As remarked in the proof of Lemma~\ref{construction-pas-sym} there is a vertex $v$ that has only one outgoing edge and its weight is a constant $c$ (as $\varphi$ is supposed to have at least one addition). We change its weight to $c_0 c$ and add a loop weighted by $c_0$ on $v$. We put a loop with weight $1$ on every other vertex than $v$ and $s$.    

Let $\{1, \dots, e+1\}$ be the vertices of $\bar{G}$ and $A$ its adjacency matrix. Let us have a closer look at cycle covers of $\bar{G}$. The cycles in $\bar{G}$ are cycles containing $s$ (which are in bijection with $s$-$t$-paths in $G$) and loops. In a cycle cover $C$ the vertex $s$ belongs to a cycle $S$. Its weight $w(s)$ is the weight of the corresponding $s$-$t$-path $P$ in $G$ and its cardinal is $|S|=|P|-1$. If the vertex $v$ appears in $S$ then $w(S)=c_0 w(P)$ and every other cycle in $C$ is a loop of weight $1$. Otherwise  $w(S)=w(P)$ and $C$ contains the loop $v$ of weight $c_0$. In both case $w(C)=c_0 w(P)$. Let us recall that $\operatorname{sgn}(C)$ is the signature of the underlying permutation: here it is $-1$ if $S$ is even and $1$ otherwise, and so it is equal to $(-1)^{|P|}$. Using Lemma~\ref{lemma:perm-det} we get
\[\det(A) = \sum_{\substack{\text{cycle cover}\\\text{$C$ of $\bar G$}} }  \operatorname{sgn}(C) w(C)
   = c_0 \cdot\sum_{\substack{\text{$s$-$t$-path}\\P\in G}}(-1)^{|P|} \ w(P) = \varphi.\]
\end{proof}

\subsection{Symmetric case} 
\label{sec:sym-form}

The aim of this section is to write an arithmetic formula as a determinant of a symmetric matrix, whose entries are constants or variables. 
Recall that in this section as in Section~\ref{sec:ws-circuits}, a field $k$ of characteristic different from $2$ is fixed, and the input constants are taken from this field. In the sequel, every constructed graph is undirected. 
At first, the result is proved for the skinny size of the formula. We recall that the skinny size of $\varphi$ is the number of arithmetic operators it contains.

\begin{theorem}
\label{circsym}
Let $\varphi$ be an arithmetic formula of skinny size $e$. Then there exists a matrix $A$ of dimensions 
at most $2e+3$ whose entries are inputs of the formula and elements of $\{0,1,-1,1/2\}$ such that $\varphi=\det A$.
\end{theorem}

This theorem is a corollary of the following lemma.

\begin{lemma}
\label{construction}
Let $\varphi$ be an arithmetic formula of skinny size $e$. Then there exists a graph $G$ with at most $2e+2$ vertices and two distinct vertices $s$ and $t$ such that
\begin{enumerate}
\item
The graph $G$ has an even number of vertices, every cycle in $G$ is even and every $s$-$t$-path has an even number of vertices.
\item
The subgraph $G\setminus \{s,t \}$ is empty if $e=0$ and for $e\geq 1$ it has only one cycle cover:
It is a perfect matching of weight $1$. For every $s$-$t$-path $P$ in $G$, the subgraph $G\setminus P$ is empty or has only one cycle cover: as above it is a perfect matching of weight $1$.
\item
The following equality holds in $G$:
\[\sum_{\text{$s$-$t$-path $P$}}(-1)^{|P|/2+1} \ w(P) = \varphi\]

\end{enumerate}
The graph $G$ is called the graph associated to $\varphi$.
\end{lemma}

The first property of the lemma ensures that because of a parity argument every cycle cover of the final constructed graph $\bar{G}$ used in the proof of Theorem~\ref{circsym} (see Fig.~\ref{fig:form-loop}) includes exactly one path between $s$ and $t$. The second property ensures that the weight of the cycle cover is the weight of the  cycle involving $s$ and $t$, that is every other cycle has weight $1$, and that other cycles of the cover are of length 2. The third property gives the relation between the graph and the formula.

As in Valiant's construction for the not 
necessarily symmetric case, the formula $\varphi$ will be encoded in the weights of paths between $s$ and $t$, but in a slightly different way. In Valiant's construction, a cycle cover of the digraph is made of a cycle including a $s$-$t$-path, other cycles being loops. Moreover every $s$-$t$-path has the same parity and so every cycle cover has the same parity of odd cycles and the underlying permutation has the same signature. With this property of the digraph the determinant of its adjacency matrix is equal to its permanent up to the sign. In our construction a cycle cover of the graph is made of a cycle including a $s$-$t$-path, other cycles being length-$2$ cycles. A length-$2$ cycle has a negative signature and every $s$-$t$-path of the graph has an even cardinality, so the sign of the cycle permutation is $-1$ to the number of length $2$ cycles. This shows that the sign of the cycle permutation is a function of the length of the involved $s$-$t$-path \emph{modulo} 4. There is a way to ensure that this sign does not depend on the chosen $s$-$t$-path: replace the graph $G$ associated to a size-$0$ formula $x$ in the proof of Lemma~\ref{construction} by a $4$-vertices path with weight $x$ on its first edge, and replace weights $-1$ (Fig.~\ref{fig:form-loop}, Fig.~\ref{fig:form-sum2} and Fig.~\ref{fig:form-prod}) by weights $1$. This yields a matrix with entries in $k\cup\bar x$ whose determinant and permanent are equal to $\varphi$, but its dimensions 
can be $4e+5$. To achieve the $2e+3$ bound, we construct a matrix $A$ whose determinant can be very different from the permanent: For example, the permanent of the matrix associated to $\varphi=x+x$ is 0 when its determinant is $2x$. Nonetheless we can very easily obtain a matrix $B$ having the same dimensions as 
$A$ and such that $\operatorname{perm} B=\varphi$ by replacing every $-1$ entry in $A$ by $1$.  

\bigskip

\begin{proof}[Proof of Theorem~\protect\ref{circsym}]
Let $G$ be the graph associated to $\varphi$ and let $\bar{G}$ be the graph $G$ augmented with a new vertex $c$ and the edges $tc$ of weight $1/2$ and $cs$ of weight $(-1)^{|G|/2-1}$ (see Fig.~\ref{fig:form-loop}). 

    \begin{figure}[htbp]
    \begin{center}
    \input{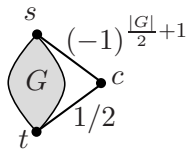}
    \caption{Construction of $\bar G$ from $G$.}
    \label{fig:form-loop}
    \end{center}
    \end{figure}

Conditions $(1)$ and $(2)$ imply that there is a bijection between paths from $s$ to $t$ or $t$ to $s$ and cycle covers in $\bar{G}$. More precisely, every cycle cover in $\bar{G}$ has a unique odd cycle and it is of the form $cPc$ where $P$ is a $s$-$t$-path or a $t$-$s$-path. Indeed, the graph $\bar{G}$ has an odd number of vertices. Suppose there is a cycle cover of $\bar{G}$ involving the length-$2$ cycle $tct$. Other cycles of this cover are cycles of $G$ and thus by $(1)$ they are all even. This is not possible as an odd set can not be partitioned into even subsets. For the same reason, there is no cycle cover of $\bar{G}$ involving the cycle $scs$. Thus every cycle cover of $\bar{G}$ has a cycle including $c$ and a path $P$ between $s$ and $t$.

Let us recall that the sign of a cycle cover is the sign of the underlying permutation, i.e. $-1$ if it has an odd number of even cycles and $1$ otherwise, and let us define the signed weight of a cycle cover as the product of its weight and sign. Let $C$ be a cycle cover of $\bar{G}$ involving the $s$-$t$-path $P$. By property $(2)$ there is only one way to complete the cover. Thus the weight of the cycle cover is the weight of $P$ multiplied by  $(1/2 \ (-1)^{|G|/2+1})$ and its sign is the sign of a perfect matching of cardinality $|G\setminus P|$, so it is $(-1)^{(|G\setminus P|)/2}$. By symmetry, the inverse cycle cover has the same signed weight. So the sum of the signed weights of all cycle covers of $\bar{G}$ is equal to twice the sum over all $s$-$t$-path $P$ of $(1/2 \ (-1)^{|P|/2+1}\ w(P))$. According to Lemma~\ref{construction} it is equal to $\varphi$. The result follows from Lemma~\ref{lemma:perm-det}.

\end{proof}

\begin{proof}[Proof of Lemma~\protect\ref{construction}]
We proceed by structural induction. In other words, we first prove the lemma for the simplest possible formula, namely $x$, and then show that the assertion of Lemma~\ref{construction} is stable under addition and multiplication. 

Let $\varphi=x$ be an arithmetic formula of size $0$. Then the graph $G$ associated to $\varphi$ by definition has two vertices $s$ and $t$ and an edge $st$ of weight $x$. 
It verifies trivially properties $(1)$ and $(2)$ and its only $s$-$t$-path is $st$ and we have: $(-1)^{2/2+1}x=\varphi$.

Let $\varphi=\varphi_1+ \varphi_2$ and $G_1$ and $G_2$ be the graphs associated to $\varphi_1$ and $\varphi_2$. First let us suppose  $s_1t_1$ or $s_2t_2$ has weight $0$. This 
means in particular that $\varphi_1$ or $\varphi_2$ is of size at least 1.
Let $s=s_1=s_2$ and $t=t_1=t_2$. Suppose $G_1\setminus \{s_1,t_1 \}$ and $G_2\setminus \{s_2,t_2 \}$ have disjoints sets of vertices and let $G=G_1\cup G_2$ (see Fig.~\ref{fig:form-sum}). Then $|G|=|G_1|+|G_2|-2\leq 2|\varphi_1|+2|\varphi_2|+2 = 2|\varphi|$.

    \begin{figure}[htbp]
    \begin{center}
    \input{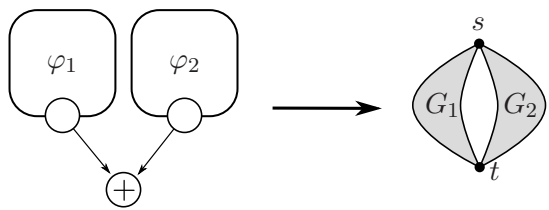}
    \caption{Graph associated to $\varphi=\varphi_1+\varphi_2$.}
    \label{fig:form-sum}
    \end{center}
    \end{figure}

If $s_1t_1$ is an edge in $G_1$ and $s_2t_2$ is an edge in $G_2$ then the preceding construction would lead to two edges between $s$ and $t$. They could be transformed into a single edge if adding the two weights, but then the weight could be a sum of two variables, and it is something that is not allowed in this context. So the graph $G_1$ is transformed into a graph $G_1'$ by adding two vertices $u$ and $v$, removing the edge $s_1t_1$ with weight $x$ and adding the edges $s_1 u$ with weight $x$, $uv$ with weight $1$ and $vt_1$ with weight $-1$ (see Fig.~\ref{fig:form-sum2}).

    \begin{figure}[htbp]
    \begin{center}
    \input{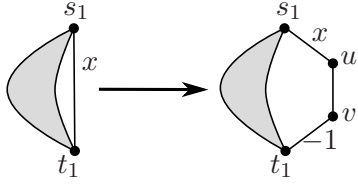}
    \caption{Transformation of $G_1$ into $G'_1$.}
    \label{fig:form-sum2}
    \end{center}
    \end{figure}
We can verify easily that $G_1'$ satisfies 
the three conditions of Lemma~\ref{construction}. In particular for the third condition, the term $x$ corresponding to the path $s_1t_1$ in $G_1$ in the sum is replaced by the term corresponding to the path $s_1uvt_1$ in $G_1'$: $-(-1)^{4/2+1}x=x$. We then construct the graph $G$ associated to $\varphi$ as above but with $G_1'$ replacing $G_1$. It size is at most $2|\varphi|+2$.

Now let us prove that the graph associated to $\varphi$ satisfies the three properties of the lemma.

\begin{enumerate}
\item
$G$ has an even number of vertices and the cardinality of every $s$-$t$-path is even. A cycle in $G$ is a cycle in $G_1$, or a cycle in $G_2$, or a path from $s$ to $t$ in $G_1$ or $G_2$ followed by path from $t$ to $s$ in $G_1$ or $G_2$, and consequently every cycle in $G$ is even.
\item
If $G_1\setminus \{s_1,t_1 \}$ and $G_2\setminus \{s_2,t_2 \}$ are non-empty they are disconnected, and a cycle cover of the subgraph $G\setminus \{s,t \}$ is constituted by a cycle cover of $G_1\setminus \{s_1,t_1 \}$ and a cycle cover of $G_2\setminus \{s_2,t_2 \}$. So $G\setminus \{s,t \}$ has only one cycle cover and it is a perfect matching of weight $1$. If $G_1\setminus \{s_1,t_1 \}$ is empty then $G\setminus \{s,t \}=G_2\setminus \{s_2,t_2 \}$ and has only one cycle cover and it is a perfect matching of weight $1$.

Let $P$ be a path between $s$ and $t$ in $G$. We can suppose wlog that the subgraph $G\setminus P$ is the union of the two graphs $G_1\setminus P$ and $G_2\setminus \{s_2,t_2 \}$, which are disconnected from one another. The property to prove is then straightforward from the induction hypothesis.
\item
A path of $G$ is a path of $G_1$ or a path of $G_2$, which proves the equality.
\end{enumerate}

Let $\varphi=\varphi_1\times \varphi_2$ and $G_1$ and $G_2$ be the graphs associated to $\varphi_1$ and $\varphi_2$. Suppose $G_1$ and $G_2$ have disjoints sets of vertices and let $G$ be $G_1\cup G_2$ with an additional edge $t_1 s_2$ of weight $-1$, and let $s=s_1$ and $t=t_2$
(see Fig.\ref{fig:form-prod}).

    \begin{figure}[htbp]
    \begin{center}
    \input{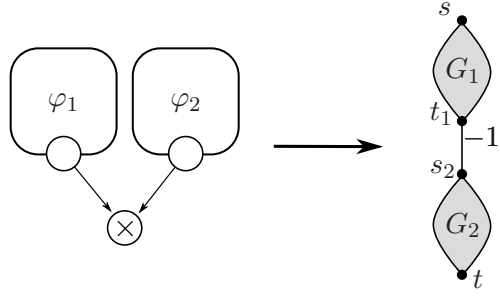}
    \caption{Graph associated to $\varphi=\varphi_1\times\varphi_2$.}
    \label{fig:form-prod}
    \end{center}
    \end{figure}
Then $|G|=|G_1|+|G_2|\leq 2|\varphi_1|+2|\varphi_2|+4 = 2|\varphi|+2$. Let us prove that $G$ satisfies the three properties of the lemma.

\begin{enumerate}
\item
$G$ has an even number of vertices and every path from $s$ to $t$ has an even cardinality. A cycle in $G$ is either a cycle in $G_1$, or a cycle in $G_2$ or the length-$2$ cycle $t_1s_2$, and consequently every cycle in $G$ is even.
\item
Let us consider a cycle cover of $G\setminus \{s,t \}$. The vertex $t_1$ can be in a cycle of $G_1$ or in the cycle $t_1s_2$. If it is in a cycle of $G_1$ then we have a cycle cover of $G_1\setminus \{s_1\}$, which is not possible because it is an odd set and all its cycles are even. Thus the cycle cover of $G\setminus \{s,t \}$ can be partitioned into $t_1s_2$ of weight $(-1)^2$, a cycle cover of $G_1\setminus \{s_1,t_1 \}$ and a cycle cover of $G_2\setminus \{s_2,t_2 \}$. Those cycle covers are unique and so there is only one cycle cover of $G\setminus\{x,y\}$ and it is a perfect matching of weight $1$.

Let $P$ be a path between $s$ and $t$ in $G$. It is a path $P_1$ from $s_1$ to $t_1$ in $G_1$ followed by $t_1s_2$ and a path $P_2$ from $s_2$ to $t_2$ in $G_2$. So $G\setminus P$ is the union of the two graphs $G_1\setminus P_1$ and $G_2\setminus P_2$, which are disconnected (if non empty) from one another. The property to prove is then straightforward from the induction hypothesis.
\item
A $s$-$t$-path $P$ in $G$ can be decomposed into three paths: a $s_1$-$t_1$-path $P_1$, $t_1s_2$ which is of weight $-1$ and a $s_2$-$t_2$-path $P_2$.

Thus  
\begin{eqnarray*}
(-1)^{\frac{|P|}{2}+1} \ w(P) &=& (-1)^{\frac{|P_1|+|P_2|}{2}+1} w(P_1) (-1) w(P_2)\\
                      &=& (-1)^{\frac{|P_1|}{2}+1} \ w(P_1) \times (-1)^{\frac{|P_2|}{2}+1} \ w(P_2)
\end{eqnarray*}  

and so

\begin{eqnarray*}
\sum_{P}(-1)^{\frac{|P|}{2}+1} \ w(P) &=& \sum_{P_1}(-1)^{\frac{|P1|}{2}+1}w(P_1) \times \sum_{P_2}(-1)^{\frac{|P_2|}{2}+1}w(P_2) \\
                                      &=& \varphi_1 \times \varphi_2 \\
                                      &=& \varphi.
\end{eqnarray*}

\end{enumerate}
\end{proof}

The upper bound $(2e+2)$ of Lemma~\ref{construction} is tight as shown by Fig.~\ref{fig:form-upper}. It can be shown easily that this construction yields a graph of size at least $|\varphi|+2$, and this lower bound is tight as shown by Fig.~\ref{fig:form-lower}.

    \begin{figure}[htbp]
    \begin{center}
    \input{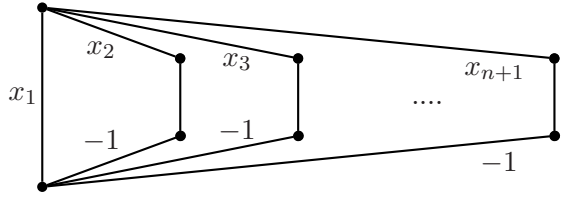}
    \caption{Graph associated to $\varphi=x_1+\cdots+x_{n+1}$: $|\varphi|=n$ and $|G|=2n+2$.}
    \label{fig:form-upper}
    \end{center}
    \end{figure}

    \begin{figure}[htbp]
    \begin{center}
    \input{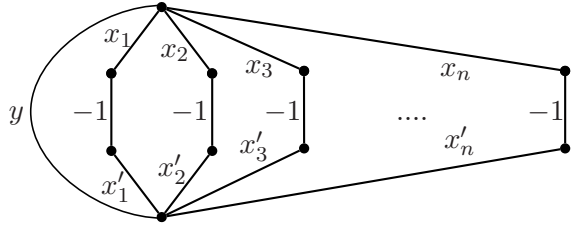}
    \caption{Graph associated to $\varphi=x_1x'_1+x_2x'_2\cdots+x_n x'_n+y$: $|\varphi|=2n$ and $|G|=2n+2$.}
    \label{fig:form-lower}
    \end{center}
    \end{figure}

In fact, as in the non-symmetric case, the skinny size can be replaced by the green size of the formula defined in Definition~\ref{def:green-size}.

\begin{theorem}
\label{greensym}
For every formula $\varphi$ of green size $e$ there is a square matrix $A$ of dimensions 
$2e+3$ whose entries are inputs of the formula and elements of $\{0,1,-1,1/2\}$ such that $\varphi=\det A$.
\end{theorem}

\begin{proof}
It is sufficient to show how to have the constants for free in the construction of Lemma~\ref{construction}. 
We also proceed by structural induction. 
In fact, the construction remains almost the same but with the last property changed. For an arithmetic formula $\varphi$ of green size $e$, there exists a graph $G$ that satisfies the conditions of Lemma~\ref{construction} but the third one is replaced by the existence of a constant $c_0$ such that
\[c_0\cdot\sum_{\text{$s$-$t$-path $P$}}(-1)^{|P|/2+1} \ w(P) = \varphi.\]

Let $\varphi=x$ be an arithmetic formula of size $0$. Then the graph $G$ associated to $\varphi$ by definition has two vertices $s$ and $t$ and an edge $st$ of weight $x$. The associated constant is $c_0=1$.

Let $\varphi=c \psi$ and $G$, $c_0$ be associated to $\psi$. Then $G$, $c c_0$ is associated to $\varphi$.

Let $\varphi=\varphi_1\times \varphi_2$ and $G_1$, $c_1$ (resp. $G_2$, $c_2$) be associated to $\varphi_1$ (resp. $\varphi_2$).
The graph $G$ associated to $\varphi$ is exactly the same as in the proof of Lemma~\ref{construction} and the constant is $c_1 c_2$.

Let $\varphi=\varphi_1+ \varphi_2$ and $G_1$, $c_1$ (resp. $G_2$, $c_2$) be the graph and constant associated to $\varphi_1$ (resp. $\varphi_2$). We suppose that $G_1$ and $G_2$ have distinct sets of vertices except for $s_1=s_2$. The graph $G$ is obtained by adding a new vertex $u$, an edge $t_2u$ with weight $1$ and an edge $ut_1$ with weight $-c_2/c_1$, and the associated constant is $c_1$ (see Fig.~\ref{fig:form-sum-cst}).

    \begin{figure}[htbp]
    \begin{center}
    \input{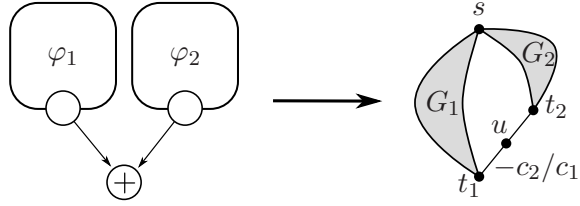}
    \caption{$\varphi=\varphi_1+ \varphi_2$; $G_1,c_1$ and $G_2,c_2$ are respectively associated to $\varphi_1$ and $\varphi_2$.}
    \label{fig:form-sum-cst}
    \end{center}
    \end{figure} 

This defines a size-$(2e+2)$ graph $G$ associated to a green size-$e$ formula $\varphi$. It remains to turn this graph into a matrix. Let $\bar{G}$ be the graph $G$ augmented with a new vertex $c$ and the edges $tc$ of weight $c_0/2$ and $cs$ of weight $(-1)^{|G|/2-1}$. The adjacency matrix $A$ of  $\bar{G}$ satisfies $\varphi=\det(A)$ and the proof is similar to the one of Theorem~\ref{circsym}.
\end{proof}

The bound obtained in Theorem~\ref{greensym} can be sharpened when $k=\mathbb R$ or $\mathbb C$. The idea is to build $\bar G$ by merging $s$ and $t$ instead of adding a new vertex. Suppose that $\varphi$ has at least one addition gate. Let $w=\sqrt{|c_0|/2}$. In the construction for this addition gate (see Fig.~\ref{fig:form-sum-cst}), multiply the weights of $t_2u$ and $ut_1$ by $w$. A cycle cover of the graph either goes through the path $t_2ut_1$, or contains the edge $ut_2$ in its perfect matching part. In both cases, its weight is multiplied by $w^2$. Now if $(-1)^{|G|/2+1}c_0/2>0$, then the graph obtained has the satisfying properties, and the new bound is $2e+1$. If it is negative, two solutions can be applied. Either $k$ is the field of complex numbers and it is sufficient to replace $w$ by $iw$ (where $i^2=-1$) to get the same bound $2e+1$. Otherwise, if $k$ is the field of real numbers, it is sufficient to add a new vertex with a loop of weight $-1$ (this corresponds to adding a new line and a new column, filled with zeroes but the diagonal element with $-1$) to get the bound $(2e+2)$.

\section{Weakly skew circuits} 
\label{sec:ws-circuits}

In this section, we extend the previous results to the case of weakly-skew circuits. Recall that those circuits are defined from arithmetic circuits by a restriction on the multiplication gate: the sub-circuit associated to one of the arguments of a multiplication gate $\alpha$ has to be closed, that is only connected to the rest of the circuit by the arrow going to $\alpha$. A gate that is not in any such closed sub-circuit is said to be \emph{reusable}.

The main difficulty to extend the results is 
the existence of several reusable gates. In the case of formulas, there is a single output. Therefore, there is a single vertex $t$ in the graph for which the sum of the weights of the $s$-$t$-paths has to equal a given expression. This is no longer the case for weakly-skew circuits. If the matrix we wish to construct is not symmetric, that is if the graph is oriented, this difficulty is overcome by ensuring that the graph is a directed acyclic graph. In that way, adding a new vertex cannot change the expressions computed at previously added vertices. But in the symmetric case, adding a new vertex, for example in the case of an addition gate, creates some new paths in the graph. Thus it changes the sum of the weights of the $s$-$t_\alpha$-paths for some vertex $t_\alpha$.

A solution to this problem is given in Lemma~\ref{lemma:sym-weak} by introducing the notion of acceptable paths: A path $P$ in a graph $G$ is said \emph{acceptable} if $G\setminus P$ admits a cycle cover. 

\subsection{Symmetric determinantal representation} 
\label{sec:fat-size}

For the weakly-skew circuits, the green size is no longer appropriate. Hence, the results of this section are expressed in terms of the fat size of the circuits: the \emph{fat size} of a circuit is its total number of gates, including the input gates. This measure of circuit size 
is refined in Section~\ref{sec:green-size}.

\begin{theorem}\label{thm:sym-weak}
Let $f$ be a polynomial computable by a weakly-skew circuit of fat size $m$. Then there exists a symmetric matrix $A$ of dimensions 
at most $2m+1$ whose entries are inputs of the circuit and elements from $\{0,1,-1,1/2\}$ such that $f=\det A$.
\end{theorem}

The proof relies on the following lemma. It applies to so-called \emph{multiple-output} weakly-skew circuits. This generalization just consists of 
circuits for which there exist several out-degree-$0$ gates.

\begin{lemma}\label{lemma:sym-weak}
Let $C$ be a multiple-output weakly-skew circuit of fat size $m$. There exists a graph $G$ with at most $2m+1$ vertices and a distinguished vertex $s$ such that $|G|$ is odd, every cycle in $G$ is even, and for every reusable gate $\alpha\in C$ there exists a vertex $t_\alpha\in G$ such that
\begin{enumerate}
\item Every $s$-$t_\alpha$-path (whether acceptable or not) has an odd number of vertices;
\item For every acceptable $s$-$t_\alpha$-path $P$ in $G$, the subgraph $G\setminus P$ is either empty or has a unique cycle cover, which is a perfect matching of weight $1$;
\item The following equality holds in $G$:
    \begin{equation}\label{eq:weak}
    \sum_{\substack{\text{acceptable}\\\text{$s$-$t_\alpha$-path $P$}} } (-1)^{\frac{|P|-1}{2}}w(P)=f_\alpha
    \end{equation}
    where $f_\alpha$ is the polynomial computed by the gate $\alpha$.
\end{enumerate}
Furthermore, the graph $G\setminus\{s\}$ has a unique cycle cover which is a perfect matching of weight $1$.
\end{lemma}

\begin{proof}
The graph $G$ is built by induction on the (fat) size of the circuit, the required properties being verified at each step of the induction. If $\alpha$ is a reusable gate of $C$, then $t_\alpha$ is said to be a reusable vertex of $G$.

A size-$1$ circuit is an input gate $\alpha$ with label $x$. The corresponding graph $G$ has three vertices: $s$, $t_\alpha$ and an additional vertex $v_\alpha$. There is an edge between $s$ and $v_\alpha$ of weight $x$, and an edge between $v_\alpha$ and $t_\alpha$ of weight $-1$. It is straightforward to check that $G$ satisfy the conditions of the lemma.

Let $m>1$ and suppose that the lemma holds for any multiple-output weakly-skew circuit of size less than $m$. Let $C$ be a multiple output weakly-skew circuit of size $m$, and $\alpha$ be any of its outputs.

If $\alpha$ is an input gate with label $x$, let $C'=C\setminus\{\alpha\}$ and $G'$ the corresponding graph  with a distinguished vertex $s$. The graph $G$ is obtained from $G'$ by adding two new vertices $v_\alpha$ and $t_\alpha$, an edge of weight $x$ between $s$ and $v_\alpha$ and an edge of weight $-1$ between $v_\alpha$ and $t_\alpha$ (see Fig.~\ref{fig:weak-input}). The vertex $s$ is the distinguished vertex of $G$. 

    \begin{figure}[htbp]
    \begin{center}
    \input{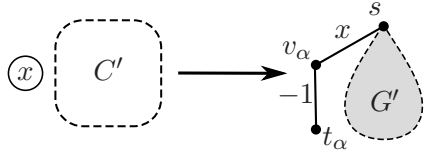}
    \caption{Induction step when $\alpha$ is an input gate.}
    \label{fig:weak-input}
    \end{center}
    \end{figure}
The size of $G$ is $|G|=|G'|+2\le (2(m-1)+1)+2=2m+1$. Thus $|G|$ is odd. A cycle in $G$ is either a cycle in $G'$ or one of the two cycles $sv_\alpha$ or $v_\alpha t_\alpha$, so every cycle in $G$ is even. The size-$3$ path from $s$ to $t_\alpha$ is acceptable (as $G'\setminus\{s\}$ has a unique cycle cover of weight $1$) and satisfies \eqref{eq:weak}. 
Now, any other reusable gate $\beta$ belongs to $C'$, so the conditions are satisfied by induction hypothesis (it is sufficient to remark that when $s$ is removed, $v_\alpha$ and $t_\alpha$ are disconnected from the rest of the circuit, and a cycle cover has to match those two vertices).

If $\alpha$ is an addition gate, let $C'=C\setminus\{\alpha\}$ and suppose that $\alpha$ receives arrows from gates $\beta$ and $\gamma$. Note that $\beta$ and $\gamma$ are reusable. Let $G'$ be the graph corresponding to $C'$, and $s$ be its distinguished vertex. $G'$ contains two reusable vertices $t_\beta$ and $t_\gamma$. The graph $G$ is obtained by adding two vertices $v_\alpha$ and $t_\alpha$, and the following edges: $t_\beta v_\alpha$ and $t_\gamma v_\alpha$ of weight $1$, and $v_\alpha t_\alpha$ of weight $-1$ (see Fig.~\ref{fig:weak-sum}). If $\beta=\gamma$, then $G'$ contains a vertex $t_\beta$, and we merge the two edges adjacent to $t_\beta$ and $t_\gamma$ into an edge $t_\beta v_\alpha$ of weight $2$.

    \begin{figure}[htbp]
    \begin{center}
    \input{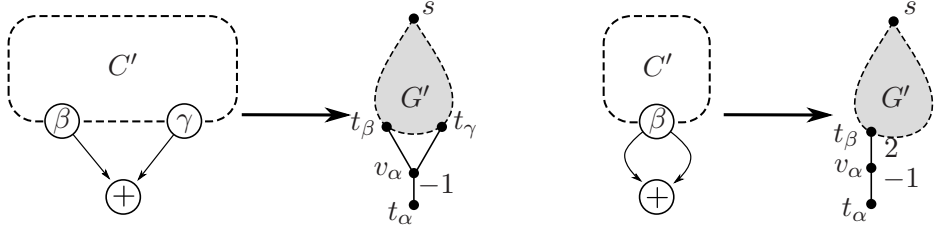}
    \caption{Induction step when $\alpha$ is an addition gate.}
    \label{fig:weak-sum}
    \end{center}
    \end{figure}
Then $|G|=|G'|+2\le 2m+1$, and $|G|$ remains odd. 

Every $s$-$t_\delta$-path for some reusable gate $\delta$ in $C'$ is even. A cycle in $G$ is either a cycle in $G'$, or the cycle $v_\alpha t_\alpha$, or is made of a 
$t_\beta$-$t_\gamma$-path $P$ in $G'$ plus the vertex $v_\alpha$. Let $P'$ be a $s$-$t_\beta$-path and $u$ the first vertex of $P'$ belonging to $P$. Then, $P'=s,\dots,u,\dots,t_\beta$ and $P''=s,\dots,u,\dots,t_\gamma$ are both path with an odd number of vertices. 
In particular the sizes of $u,\dots,t_\beta$ and $u,\dots,t_\gamma$ are of same parity. Thus $P$ is of odd size and $P\cup\{v_\alpha\}$ is an even-size cycle. Hence, every cycle in $G$ is even. An acceptable path in $G$ is either an acceptable path in $G'$ or a path from $s$ to $t_\alpha$. Indeed, the only way to cover $t_\alpha$ in a cycle cover is to match it with $v_\alpha$. Therefore, no acceptable path goes through $t_\beta$, $v_\alpha$ and $t_\gamma$. So, the reusable gates in $C'$ satisfy the conditions of the lemma by induction. Any acceptable path $P$ from $s$ to $t_\alpha$ is an acceptable path $P'$ from $s$ to $t_\beta$ or $t_\gamma$ followed by a path from $t_\beta$ or $t_\gamma$ to $t_\alpha$. Thus $|P|=|P'|+2$ is odd and $G\setminus P=G'\setminus P'$ has a unique cycle cover which is a perfect matching of weight $1$. Finally,
\begin{align*}
 &\sum_{\substack{\text{acceptable}\\\text{$s$-$t_\alpha$-path $P$}} } (-1)^{\frac{|P|-1}{2}}w(P)\\
=&\sum_{\substack{\text{acceptable}\\\text{$s$-$t_\beta$-path $P_\beta$}} } (-1)^{\frac{|P_\beta|+2-1}{2}}(-1\cdot w(P_\beta))
+ \sum_{\substack{\text{acceptable}\\\text{$s$-$t_\gamma$-path $P_\gamma$}} } (-1)^{\frac{|P_\gamma|+2-1}{2}}(-1\cdot w(P_\gamma))\\
=&\sum_{P_\beta} (-1)^{\frac{|P_\beta|-1}{2}} w(P_\beta) + \sum_{P_\gamma} (-1)^{\frac{|P_\gamma|-1}{2}} w(P_\gamma)\\
=&f_\beta+f_\gamma = f_\alpha.
\end{align*}

If $\alpha$ is a multiplication gate, $\alpha$ receives arrows from two distinct gates $\beta$ and $\gamma$. Exactly one of those gates, say $\beta$, is not reusable and removing the gate $\alpha$ yields two disjoint circuits $C_1$ and $C_2$ (say $\beta$ belongs to $C_1$ and $\gamma$ to $C_2$). Let $G_1$ and $G_2$ be the respective graphs obtained by induction from $C_1$ and $C_2$, with distinguished vertices $s_1$ and $s_2$ respectively. The graph $G$ is obtained as in Fig.~\ref{fig:weak-prod} as the union of $G_1$ and $G_2$ where $t_\gamma$ and $s_1$ are merged, the distinguished vertex $s$ of $G$ being the distinguished vertex $s_2$ of $G_2$, and $t_\alpha$ being equal to $t_\beta$. 

    \begin{figure}[htbp]
    \begin{center}
    \input{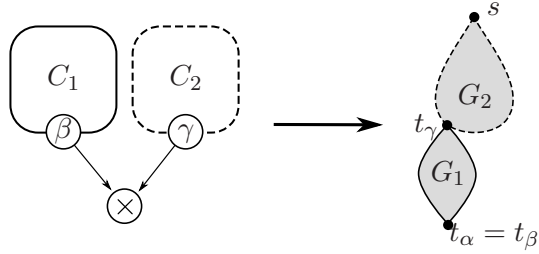}
    \caption{Induction step when $\alpha$ is a multiplication gate.}
    \label{fig:weak-prod}
    \end{center}
    \end{figure}
Then $|G|=|G_1|+|G_2|-1$, so $|G|$ is odd, and if $m_1$ and $m_2$ are the respective sizes of $C_1$ and $C_2$ ($m=m_1+m_2+1$), then $|G|\le 2m_1+1+2m_2+1-1=2m-1$. A cycle in $G$ is either a cycle in $G_1$ or a cycle in $G_2$ and is therefore even. The reusable gates of $C$ are $\alpha$ and the reusable gates of $C_2$ (by definition, $C_1$ is closed and in particular $t_\beta$ is not reusable). A path (in $G$) from $s$ to a reusable gate of $G_2$ cannot enter $G_1$ so the reusable gates of $G_2$ satisfy the first and the third conditions in the lemma. Furthermore, if such a path $P$ is removed from $G$, the only cycle cover of $G\setminus P$ has to be made of a cycle cover of $G_2\setminus P$ and a cycle cover of $G_1\setminus s_1$. Indeed, the vertex $s_1=t_\gamma$ has to be either in a cycle cover of $G_1$ or in a cycle cover of $G_2$. But $G_2\setminus(P\cup\{t_\gamma\})$ is a graph of odd size and cannot be covered by cycles of even size and $G_1$ is also of odd size. Thus, the reusable gates in $G_2$ also satisfy the second condition of the lemma.
It remains to prove that the reusable gate $\alpha$ satisfies the conditions of the lemma:
\begin{enumerate}
\item A $s$-$t_\alpha$-path $P$ is a $s$-$t_\gamma$-path $P_\gamma$ followed by a $s_1$-$t_\beta$-path $P_\beta$. Thus $|P|=|P_\gamma|+|P_\beta|-1$ as     $t_\gamma=s_1$ and $|P|$ is odd.
\item The graph $G\setminus P$ is the disjoint union of $(G_2\setminus P_\gamma)$ and $(G_1\setminus P_\beta)$, so by induction $G\setminus P$ is either     empty or has a unique cycle cover which is a perfect matching of weight $1$.
\item As $w(P)=w(P_\gamma)w(P_\beta)$, we have
    \begin{align*}
    (-1)^{\frac{|P|-1}{2}}w(P)&= (-1)^{\frac{|P_\gamma|+|P_\beta|-2}{2}}w(P_\gamma)w(P_\beta)\\
    &=(-1)^{\frac{|P_\gamma|-1}{2}}w(P_\gamma)\times (-1)^{\frac{|P_\beta|-1}{2}}w(P_\beta),
    \end{align*}
    whence
    \begin{align*}
    \sum_{P}(-1)^{\frac{|P|-1}{2}} w(P) &= \sum_{P_\gamma}(-1)^{\frac{|P_\gamma|-1}{2}}w(P_\gamma)
        \times\sum_{P_\beta}(-1)^{\frac{|P_\beta|-1}{2}}w(P_\beta) \\
    &= f_\gamma\times f_\beta \\
    &= f_\alpha.
    \end{align*}  
\end{enumerate}
Finally, the only way to cover $G\setminus\{s\}$ is to cover $G_2\setminus\{s_2\}$ on one hand and $G_1\setminus\{s_1\}$ on the other hand for parity reasons as before. The weight of this cover is the product of the weights of the covers of $G_1$ and $G_2$, that is $1$.
\end{proof}

\begin{proof}[Proof of Theorem \protect\ref{thm:sym-weak}]
Let $C$ be a weakly-skew circuit computing the polynomial $f$, and $G$ be the graph built from $C$ in Lemma~\ref{lemma:sym-weak}. The circuit $C$ has a unique output, and there exists in $G$ a vertex $t$ corresponding to this output. Let $G'$ be the graph obtained from $G$ by adding an edge between $t$ and $s$ of weight $\tfrac12(-1)^{\frac{|G|-1}{2}}$.

There is no cycle cover of $G'$ containing the $2$-cycle $st$. Indeed, $\left|G'\setminus\{s,t\}\right|$ is odd and $G$ contains only even cycles. This means that a cycle cover of $G'$ contains a cycle made of a $s$-$t$-path plus $(t,s)$ or a $t$-$s$-path plus $(s,t)$. Let $P$ be such a path. Then $G'\setminus P=G\setminus P$. Hence, by Lemma~\ref{lemma:sym-weak}, there is exactly one cycle cover of $G'\setminus P$ and it is a perfect matching of weight $1$. This means that there is a one-to-one correspondence between the cycle covers of $G'$ and the paths from $s$ to $t$ or from $t$ to $s$. There is also a one-to-one correspondence between the paths from $s$ to $t$ and the paths from $t$ to $s$. 

Let us recall that the sign of a cycle cover is the sign of the underlying permutation and its signed weight is the product of its sign and weight. Let $C$ be a cycle cover of $G'$ involving the $s$-$t$-path $P$. The previous paragraph shows that the weight of $C$ equals $\tfrac12(-1)^{\frac{|G|-1}{2}}w(P)$. As $C$ has an odd cycle and a perfect matching, its sign is $(-1)^{|G\setminus P|/2}$, that is the number of couples in the perfect matching. The inverse cycle cover $\bar C$ of $G'$ has the same signed weight as $C$. Hence the sum of the signed weights of all cycle covers of $G'$ equals twice the sum over all $s$-$t$-paths $P$ of $\tfrac12(-1)^{\frac{|G|-1}{2}}(-1)^{\frac{|G\setminus P|}{2}}w(P)=\tfrac12(-1)^{\frac{|P|-1}{2}}w(P)$. By Lemma~\ref{lemma:sym-weak}, this equals $f$ and Lemma~\ref{lemma:perm-det} concludes the proof.

\end{proof}

\subsection{Symmetric determinantal representation of the determinant}
\label{sec:Meena}

Let us denote by 
$\DET_n$ the formal determinant of the $n\times n$ matrix $(x_{i,j})$. This polynomial 
has a weakly-skew circuit of size-$O(n^5)$ (\citep{Ber84,MP08}) or even $O(n^4)$ if we use algebraic branching program 
constructed by Mahajan and Vinay~\citeyearp{MV97}.
This weakly-skew circuit can be represented by a determinant of a symmetric matrix as proved in this paper in Theorem~\ref{thm:sym-weak}.

After a talk from one of us presenting our results, Meena Mahajan and Prajakta Nimbhorkar have communicated us the following theorem, which shows that for the determinant polynomial, the symmetrization can be done more efficiently that in the general case. As this result is not published, we find interesting to give here its proof.

\begin{theorem}[Meena Mahajan and Prajakta Nimbhorkar]
For every $n$ there is a symmetric matrix $M$ of dimensions 
$4n^3+7$ and entries in $\{x_{i,j} \mid 1\leq i,j\leq n\}\cup\{0; 1; -1; 1/2 \}$ such that 
$\DET_n=\det M$.
\end{theorem} 

\begin{proof}
Construct the weighted graph $G'$ computing $\DET_n$ with the method used in Section~3 of \citep{MV97}. 
It is a directed acyclic weighted graph with three distinguished vertices $s$, $t_+$ and $t_-$.  Every weight is $0$, $1$ or a variable $x_{i,j}$. The graph satisfies 
$$\DET_n=\sum_{s-t_+-path \ P \text{ in } G'} w(P) - \sum_{s-t_--path \ P \text{ in } G'} w(P)$$
Moreover, this graph has $2n^3+3$ vertices, at most $4n^4$ edges and the following nice structure: it is made of $n+1$ layers, the first layer being $s$ and the last one being $\{t_+, t_-\}$. Every edge is from a layer $i$ to a layer $i+1$. As a consequence, every $s$-$t_+$-path has $n+1$ vertices, and so has every $s$-$t_-$-path.

From  the graph $G'$ we can easily obtain an algebraic branching program $G$ for computing $\DET_n$: add a vertex $t$, an edge $(t_+,t)$ of weight $1$ and an edge $(t_-,t)$ of weight $-1$. We could then proceed to built from this algebraic branching program a skew-circuit of size $O(n^4)$ (see for example proof of Proposition~1 in \citep{KaKoi08}) and then a symmetric determinantal representation of dimensions 
$O(n^4)$ with the method described in Theorem~\ref{thm:sym-weak}. But symmetrizing directly the algebraic branching program $G$ allows us to achieve a better bound as we are going to see.

Let $V$ be the set of vertices of $G\setminus \{s,t\}$ and $E$ be the set of edges of $G$. The symmetric weighted graph $G_s$ is defined as follows by duplicating vertices in graph $G$. The set of its vertices is $\{s_{out}, t_{in}\}\cup\{ u_{in}, u_{out} \mid u\in V \}$. The set of its edges is $ \{ u_{out} v_{in} \mid (u,v) \in E \} \cup \{ u_{in} u_{out} \mid u \in V \}$. Weights on edges are defined by $w(u_{out} v_{in})=w(u,v)$ and $w(u_{in} u_{out})=1$. The graph $G_s$ has 
$4n^3+6$ vertices arranged in 
$2n+2$ layers and satisfies the following property:
\begin{equation}\label{eq:strechedABP}
    \DET_n=\sum_{\substack{\text{acceptable}\\\text{$s_{out}$-$t_{in}$-path $P$ in $G_s$}} } w(P)
\end{equation}
Recall that a path $P$ in a graph $G_s$ is called 
\emph{acceptable} if $G_s\setminus P$ admits a cycle cover.

To prove Property~\ref{eq:strechedABP}, let us have a look at some acceptable $s_{out}$-$t_{in}$-path $P_s$ in $G_s$ and at some cycle cover $C$ of $G_s\setminus P_s$. We prove that for every $u\in V$, the vertices $u_{in}$ and $u_{out}$ are both in $P_s$ or together in a length-2 cycle of $C$. 
The first vertex of the path $P_s$ is  $s_{out}$. The second vertex is some $u_{1,in}$ where $u_1$ is a vertex of the second layer of $G$. The third vertex is $u_{1,out}$ as $u_{1,in}$ is only linked to $s_{out}$ and $u_{1,out}$. Let us now consider another vertex $v_{in}$ where $v$ belongs to the second layer of $G$. It is only linked to $s_{out}$ and $v_{out}$, and so it is not in $P_s$ but belongs to the weight $1$ and length-2 cycle $v_{in}v_{out}$ in $C$. The same reasoning applies to the following layers. Thus we just proved that there is a weight-preserving bijection between acceptable $s_{out}$-$t_{in}$-paths in $G_s$ and $s$-$t$-paths in $G$. Moreover, for every acceptable $s_{out}$-$t_{in}$-paths $P_s$ in $G_s$, the graph $G_s\setminus P$ has only one cycle cover, which is of weight $1$ and sign $(-1)^{|G_s\setminus P_s|/2}=(-1)^{2n^3-n+2}=(-1)^n$. Because of the symmetry of the graph we also have: 

\begin{equation}
    \DET_n=\sum_{\substack{\text{acceptable}\\\text{$t_{in}$-$s_{out}$-path $P$ in $G_s$}} } w(P)
\end{equation}

and thus

\begin{equation}
\label{eq:path5}
    \DET_n=\frac{1}{2} \sum_{\substack{\text{acceptable}\\\text{$s_{out}$-$t_{in}$-path $P$ in $G_s$}} } w(P)
    +\frac{1}{2} \sum_{\substack{\text{acceptable}\\\text{$t_{in}$-$s_{out}$-path $P$ in $G_s$}} } w(P)
\end{equation}

Remark that every cycle in $G_s$ is even because of its layer structure. Let $\bar{G}$ be the graph $G_s$ augmented with a new vertex $c$ and the edges $t_{in}c$ of weight $1/2$ and $cs_{out}$ of weight $(-1)^n$, and let $M$ be its adjacency matrix. 
The end of the proof is similar to the one of Theorem~\ref{circsym}. The only odd cycles in $\bar{G}$ are the ones including $c$ and a $s_{out}$-$t_{in}$-path or a $t_{in}$-$s_{out}$-path $P$. As an odd graph can not be decomposed in even cycles, every cycle decomposition of in $\bar{G}$ has one of these odd cycles. It was proven above that the rest of the graph has only one possible cycle decomposition. Thus
by \eqref{eq:path5}:

\begin{equation}
    \DET_n=\sum_{\text{cycle cover $C$ in $\bar{G}$}}  \operatorname{sgn}(C)\  w(C)
\end{equation}

According to Lemma~\ref{lemma:perm-det} we have
 
\begin{equation}
    \det(M)=\sum_{\text{cycle cover $C$ in $\bar{G}$}} \operatorname{sgn}(C)\  w(C)
\end{equation}

and thus the result
 
\begin{equation}
    \DET_n=\det(M)
\end{equation}

\end{proof}

\subsection{Minimization} 
\label{sec:green-size}

The aim of this section is to refine the bound we obtained in Section~\ref{sec:fat-size}, using the notion of green size that was defined in Section~\ref{sec:valiant} (and matches the notion of size used in \citep{LR06}). As mentioned before, one can refine this notion of green size. 
It relies on the idea already mentioned by Liu and Regan for the formulas: One can add weights on the arrows of the circuit. If there is an arrow from a gate $\alpha$ to a gate $\beta$ with weight $c$, then $\beta$ receives as argument the value $c f_\alpha$ where $f_\alpha$ is the polynomial computed by $\alpha$. Such a circuit is called a \emph{weighted circuit}. Of course, a classical circuit is a weighted one with all weights equal to $1$.

To refine the notion of green size, the idea is to avoid counting the variable-free sub-circuit. The next lemma shows that it is possible to do this in a very simple way.

\begin{lemma}\label{lemma:minimization}
If $C$ is a weighted circuit, then there exists an equivalent weighted circuit $C'$ with the same number of inputs labelled by a variable and at most the same number of computation gates such that:
\begin{enumerate}
\item An input gate is labelled either by a variable or the constant $1$, and the constant inputs have out-degree $1$;
\item An addition gate has at most one constant argument and this argument is an input gate;
\item A multiplication gate has both arguments non-constant.
\end{enumerate}
\end{lemma}

\begin{proof}
One can suppose that there exists some input gate labelled by a variable, otherwise the polynomial computed by $C$ would be constant. To obtain the three points, each of the four following rules is recursively applied to $C$. Each rule is applied as long as possible before we apply the next one. We never go back to a previous rule.
\begin{enumerate}
\item Every input gate labelled by a constant $c$ is replaced by an input gate labelled by $1$, and the weight of an arrow going from it is multiplied by $c$. If there are several arrows going from this input gate, it is duplicated so that each copy has out-degree $1$.
\item Every computation gate $\alpha$ that has both arguments constant is replaced by an input gate labelled by $1$, and the weight of every arrow going from it is multiplied by the value $\alpha$ computed. As in previous step, the new input gates are duplicated to have out-degree $1$.
\item If a multiplication gate $\alpha$ with positive out-degree has one constant argument $\beta$ labelled by $1$ and with an arrow from $\beta$ to $\alpha$ of weight $c_1$, and another argument $\gamma$, non-constant, with an arrow of weight $c_2$, then $\alpha$ and $\beta$ are deleted, and every arrow going from $\alpha$ of weight $c$ is replaced by an arrow going from $\gamma$ of weight $cc_1c_2$ (see Fig.~\ref{fig:min-prod}).
    
    \begin{figure}[htbp]
    \begin{center}
    \input{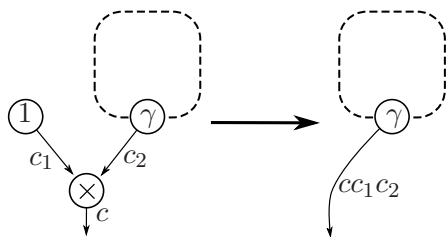}
    \caption{Minimization for a multiplication gate.}
    \label{fig:min-prod}
    \end{center}
    \end{figure}
\item If the output gate $\alpha$ is a multiplication with one constant argument $\beta$ with an arrow of weight $c_1$ going from $\beta$ to $\alpha$ and the other argument $\gamma$, non-constant, with an arrow from $\gamma$ to $\alpha$ of weight $c_2$, then $\alpha$ and $\beta$ are deleted, $\gamma$ becomes the new output gate, and the weight of every arrow coming to $\gamma$ is multiplied by $c_1c_2$ (see Fig.~\ref{fig:min-out}).
    
    \begin{figure}[htbp]
    \begin{center}
    \input{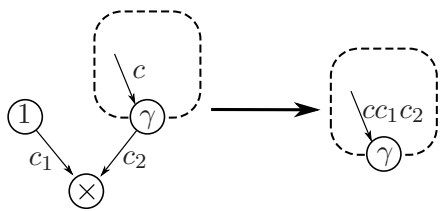}
    \caption{Minimization for the output gate.}
    \label{fig:min-out}
    \end{center}
    \end{figure}
\end{enumerate}
The first two rules ensure that all the constant input gates are labelled by $1$ and have out-degree $1$. After the second rule, each computation gate has at most one constant argument, and that it is an input gate. Then rules 3 and 4 delete all multiplication gates that have a constant argument.
\end{proof}

Note that the above lemma is valid for any kind of arithmetic circuit, and that the construction does not change the nature of the circuit. So this can be applied to a formula to get a formula, or to a weakly-skew circuit to get a weakly-skew circuit.

\begin{definition}
Let $C$ be an arithmetic circuit. Then the circuit $C'$ obtained in Lemma~\ref{lemma:minimization} is the \emph{minimized} circuit associated to $C$, and written $\min(C)$. The \emph{green size} of $C$ is equal to the skinny size of $\min(C)$, that is the number of computation gates in $\min(C)$.
\end{definition}

Note that this definition does not exactly match Definition~\ref{def:green-size} in the case of formulas, but is equivalent to the size mentioned right after the definition. 
In fact, the way of defining the green size we use here yields a smaller size. Nevertheless, it is easy to see that the results obtained in Section~\ref{sec:sym-form} remain true with this new definition.

\begin{theorem}\label{thm:weak+csts}
Let $f$ be a polynomial computable by a weighted weakly skew circuit of green size $e$ and with $i$ inputs labelled by a variable. Then there exists a symmetric matrix $A$ of dimensions 
at most $2(e+i)+1$ whose entries are inputs of the circuit and elements of $\{0,1,-1,1/2\}$ such that $f=\det A$.
\end{theorem}

\begin{proof}
The first step is to use Lemma~\ref{lemma:minimization} to minimize the circuit. Thus in the sequel the circuit is supposed to be a minimized weighted weakly-skew circuit. It is sufficient to show how to manage the constants in the construction of Lemma~\ref{lemma:sym-weak}.

The idea is to have the same construction as in Lemma~\ref{lemma:sym-weak} but with the last property changed: for every reusable gate $\alpha$, there exists a constant $c_\alpha$ such that
\begin{equation}\label{eq:with-csts}
c_\alpha\cdot\sum_{\substack{\text{acceptable}\\\text{$s$-$t_\alpha$-path $P$}} } (-1)^{\frac{|P|-1}{2}}w(P)=f_\alpha.
\end{equation}
The changes in the construction only concern the induction steps for computation gates (that is for multiplication and addition gates).

Suppose that $\alpha$ is an addition gate with one constant argument, say $\beta$, with an arrow from $\beta$ to $\alpha$ of weight $c_1$. Suppose the second argument of $\alpha$ is a non-constant gate $\gamma$ with an arrow from $\gamma$ to $\alpha$ of weight $c_2$. By induction, there exists a graph $G_\gamma$ of size $2((e-1)+i)+1$ that satisfies the conditions. In particular, there exists a distinguished vertex $s$, and a vertex $t_\gamma$ with the required properties (let $c_\gamma$ be the associated constant). Then $G$ is obtained by adding two new vertices $v_\alpha$ and $t_\alpha$ and the following edges: an edge $t_\gamma v_\alpha$ of weight $c_2c_\gamma$, an edge $v_\alpha t_\alpha$ of weight $-1$, and an edge $sv_\alpha$ of weight $c_1$ (see Fig.~\ref{fig:weak-csts}). One can check that $G$ satisfies the required properties. In particular, $t_\alpha$ satisfies \eqref{eq:with-csts} with the constant $1$, and $|G|=|G_\gamma|+2=2((e-1)+i)+1+2=2(e+i)+1$.

    \begin{figure}[htbp]
    \begin{center}
    \input{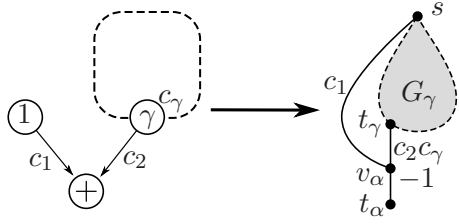}
    \caption{Graph obtained for the sum of a constant and a sub-circuit.}
    \label{fig:weak-csts}
    \end{center}
    \end{figure}

Suppose that $\alpha$ is an addition gate, receiving arrows from non-constant gates $\beta$ and $\gamma$. There exist constants $c_\beta$ and $c_\gamma$ such that \eqref{eq:with-csts} holds for $\beta$ and $\gamma$. Suppose that the arrows from $\beta$ and $\gamma$ to $\alpha$ have respective weights $c_1$ and $c_2$. The construction for the induction step in the same as in the proof of Lemma~\ref{lemma:sym-weak}, on Fig.~\ref{fig:weak-sum}, with the following changes: the edges $t_\beta v_\alpha$ and $t_\gamma v_\alpha$ are respectively weighted $c_\beta c_1$ and $c_\gamma c_2$. Note that this does not change the weight of the perfect matching as those edges never belong to those matchings. As in that case, $f_\alpha=c_1 f_\beta + c_2f_\gamma$, we obtain
\begin{align*}
 &\sum_{\substack{\text{acceptable}\\\text{$s$-$t_\alpha$-path $P$}} } (-1)^{\frac{|P|-1}{2}}w(P)\\
=&\sum_{\substack{\text{acceptable}\\\text{$s$-$t_\beta$-path $P_\beta$}} } (-1)^{\frac{|P_\beta|+2-1}{2}}(-c_1c_\beta\cdot w(P_\beta))
+ \sum_{\substack{\text{acceptable}\\\text{$s$-$t_\gamma$-path $P_\gamma$}} } (-1)^{\frac{|P_\gamma|+2-1}{2}}(-c_2c_\gamma\cdot w(P_\gamma))\\
=&c_1\cdot\biggl(c_\beta\cdot\sum_{P_\beta} (-1)^{\frac{|P_\beta|-1}{2}} w(P_\beta)\biggr)
+ c_2\cdot\biggl(c_\gamma\cdot\sum_{P_\gamma} (-1)^{\frac{|P_\gamma|-1}{2}} w(P_\gamma)\biggr)\\
=&c_1f_\beta+c_2f_\gamma = f_\alpha.
\end{align*}
Note that the constant $c_\alpha$ associated to $t_\alpha$ is equal to $1$ in that case. If $\beta=\gamma$, with the same notations as above, it is sufficient to replace the weight-$2$ edge $t_\beta v_\alpha$ by an edge of weight $2c_1c_\beta$.

In the case of a multiplication gate, the construction (shown in Fig.~\ref{fig:weak-prod}) has no available edge to put the constants. But here, if the arrows from $\beta$ and $\gamma$ to $\alpha$ are still labelled by $c_1$ and $c_2$ respectively, then $f_\alpha=c_1c_2 f_\beta f_\gamma$. Thus, the same construction is kept, and the constant $c_\alpha$ associated to $\alpha$ is defined to be $c_\alpha=c_1c_2c_\beta c_\gamma$ (where $c_\beta$ and $c_\gamma$ are respectively associated to $\beta$ and $\gamma$).

It remains to adapt the proof of Theorem~\ref{thm:sym-weak} to this case. This is easily done by multiplying the weight of the edge between $s$ and $t$ by the constant associated to the output gate.
\end{proof}

\section{Comparison with Quarez's results} 
\label{sec:comp}

In this section, a comparison between our results and those in \citep{Quarez08} is made. While Quarez builds matrices of fixed dimensions (depending only on the degree of the polynomial and its number of variables), we build matrices whose dimensions are polynomial in the size of the input formula or weakly-skew circuit. Consequently, if a polynomial can be represented as a formula or a weakly-skew circuit of small size (say polynomial in the number of variables and in the degree), then our constructions yield much smaller matrices than Quarez's. This is for example the case for the determinant polynomial (that is the determinant of a matrix of indeterminates) which is known to have a polynomial size weakly-skew circuit, or of the polynomial defined as the sum of all possible monomials of degree at most $d$ (for this, see below). On the other hand, some polynomials are not known to have such polynomial size formulas or weakly-skew circuits. A famous example among those is the permanent. We shall see that our constructions also yield better bounds in that interesting case. In the most general case though, our constructions may yield bigger matrices. The next theorem quantifies this.

\begin{theorem}\label{thm:comp}
Let $p$ be a degree-$d$ polynomial in $n$ variables over a field $k$ of characteristic different from $2$. Then $p$ admits a formula of skinny size 
\[F(n,d)\le\binom{n+d+1}{n+1}-\binom{n+d-1}{n+1}-2.\]
This yields a symmetric determinantal representation of dimensions 
\[S(n,d)\le 4\binom{n+d-1}{n}-2.\]
\end{theorem}

\begin{proof}

Let $P_{n,d}$ a degree-$d$ polynomial in $n$ variables $\{x_1,\dots,x_n\}$. We shall build a \emph{weighted formula} in the sense of Section~\ref{sec:green-size}, that is a formula with inputs in $\{1,x_1,\dots,x_n\}$ and with weights on the wires. 
We will first give an algorithm to build such a formula, and then derive an upper bound on the size of the formula so constructed.

In order to clarify the construction, let us homogenize the polynomial $P_{n,d}$ with a new variable $x_0$. There exists two homogeneous polynomials $P_{n,d-1}$ and $P_{n-1,d}$ such that $P_{n,d-1}$ is a polynomial of degree at most $(d-1)$ in $(n+1)$ variables and $P_{n-1,d}$ is a polynomial of degree at most $d$ in variables $x_0,\dots,x_{n-1}$ which satisfy
\begin{equation} P_{n,d}=x_n\cdot P_{n,d-1}+P_{n-1,d}. \label{eq:rec}\end{equation}
Along with the equations $P_{k,1}=a_0x_0+a_1x_1+\cdots+a_kx_k$ and $P_{0,\delta}=p_0x_0^\delta$, this gives a formula for the polynomial $P_{n,d}$. Clearly, some $P_{k,\delta}$ may be the zero polynomial.

The rest of the proof is devoted to compute a bound on the size of the formula obtained by Equation~\eqref{eq:rec}. Let $F(n,d)$ denote the bound on the size of the formula computing $P_{n,d}$: $F(n,d)\le F(n-1,d)+F(n,d-1)+2$. For the base cases, $F(k,1)\le k$ for all $k$, $F(0,\delta)\le\delta-1$. Let $G(N,d)=F(N-d-1,d)+2$ (for $N>d$ and $d\ge 1$). Then $G(N,d)$ satisfies Pascal's formula
\begin{equation} G(N,d)\le G(N-1,d)+G(N-1,d-1) \label{eq:pascal} \end{equation}
and $G(\delta+1,\delta)\le \delta+1$, $G(k+2,1)\le k+2$. 
Thus $G(N,d)$ is bounded from above by the binomial coefficient $\tbinom Nd$, so we obtain 
\begin{equation} F(n,d)\le\binom{n+d+1}{d}-2. \label{eq:fnd} \end{equation}
This gives a first bound on $F(n,d)$, somewhat bigger than the one announced. 
This comes from the fact that the base case bound $F(0,\delta)\le\delta-1$ is too large: As the new variable $x_0$ is for homogenization, the actual formula is obtained by replacing it by $1$ and therefore the formula for $P_{0,\delta}=p_0x_0^\delta$ is made of a single input labelled by $1$ with the constant $p_0$ on the wire going from it. So $F(0,\delta)=0$.

This remark yields the same equation as Equation~\eqref{eq:pascal} for $G$ but with a new base case $G(\delta+1,\delta)=2$. A general form for such recurrences is  
\[ G(N,d)=\sum_{j=0}^d a_j\binom N{d-j}\]
for some $a_j$. Nevertheless, the values we get for the $a_j$ if we apply this equation to the base cases are not really explicit. Therefore, we shall proceed in a different way: the new bound for $G(N,d)$ is computed as the difference between the bigger bound $\tbinom Nd$ and the number of $P_{0,\delta}$ that were counted. In the recurrence~\eqref{eq:rec}, consider the recursion tree: Suppose that the vertex corresponding to $P_{n,d-1}$ is the left child of the vertex corresponding to $P_{n,d}$, and $P_{n-1,d}$ its right child. The root of the recursion tree corresponds to the output of the formula, and its leaves to some $P_{k,1}$ or some $P_{0,\delta}$. The quantity to count is the number of leaves corresponding to some $P_{0,\delta}$. A path from the root $P_{n,d}$ to $P_{0,\delta}$ has to decrease the first argument from $n$ to $0$ and the second from $d$ to $\delta$. In the recursion tree, this corresponds to a path going $n$ times to the right and $(d-\delta)$ times to the left. Moreover, such a path finishes by a move from $P_{1,\delta}$ to its right child $P_{0,\delta}$, as $P_{0,\delta+1}$ has no child. Let us define the set of strings $W_{i,j}$ as
\[W_{i,j}=\left\{w\in\{L,R\}^* : |w|_R=i\text{ and }|w|_L=j\right\}.\]
The cardinality of $W_{i,j}$ is $\tbinom{i+j}{i}$ as an element of this set is determined by the $i$ places for the letters $R$ in a length-$(i+j)$ word. As the path from $P_{n,d}$ to $P_{0,\delta}$ finishes by a right move, the number of $P_{0,\delta}$ occurring in the recursion tree is equal to the cardinality of $W_{n-1,d-\delta}$, that is $\tbinom{n+d-\delta-1}{n-1}$. And for each $P_{0,\delta}$, the original bound counted $(\delta-1)$ operations instead of zero. Thus, to get a tighter bound we have to subtract 
\[\sum_{\delta=1}^d (\delta-1)\binom{n+d-\delta-1}{n-1}=\sum_{j=0}^{d-1} (d-j-1)\binom{n+j-1}{j}.\]
Let $\Mon_n^j$ (resp. $\Mon_n^{\le j}$) be the set of all monomials in $n$ variables of degree $j$ (resp. at most $j$). Then $\Mon_n^j$ has cardinality $\tbinom{n+j-1}{j}$, and $(d-j-1)\tbinom{n+j-1}{j}$ is the cardinality of the set $\{x^p\Mon_n^j : 0\le p\le d-j-2\}$ where $x$ is a fresh variable. Thus, the sum over $j$ of those quantities is the cardinality of $\Mon_{n+1}^{\le d-2}$, that is $\tbinom{n+d-1}{n+1}$. This gives the first part of the theorem:
\[F(n,d)\le\binom{n+d+1}{n+1}-\binom{n+d-1}{n+1}-2.\]
In the rest of the proof, we shall give a bound on the dimensions 
of the matrix obtained by our construction of Section~\ref{sec:formulas}. 

In \citep{Quarez08}, the symmetric matrix that is built contains linear functions as entries (and not only variables and constants). Therefore, we now give a bound in that case 
to permit a tighter comparison between both methods. In other words, we 
suppose that the inputs of the formula are not only constants and variables, but also linear functions. 
This amounts to defining the size of the arithmetic formula $a_0x_0+a_1x_1+\cdots+a_kx_k$ as $0$ instead of $k$. As in the previous paragraph, a direct computation where the bounds on the base cases are changed can be done but yields non explicit formulas. Therefore, we use the same technique as before: The size of the formula when inputs can be linear functions is the difference between the size of the classical formula and the number of linear functions that appear. Those linear functions are the $P_{k,1}$ and appear as leaves in the recursion tree. A leaf labelled by $P_{k,1}$ is reachable by a path going $(n-k)$ times to the right and $(d-1)$ times to the left. As above, the path finishes by a move from $P_{k,2}$ to its left child $P_{k,1}$. Therefore the number of leaves labelled by $P_{k,1}$ is the cardinality of $W_{n-k,d-2}$, that is $\tbinom{n+d-k-2}{n-k}$. All those leaves count for $k$ additions, thus the total number of saved additions is
\[\sum_{k=1}^n k \binom{n+d-k-2}{n-k}=\sum_{j=0}^{n-1}(n-j)\binom{j+d-2}{j}.\]
The computation is now the same as above and this sum equals $\tbinom{n+d-1}{d}$. Using now Theorem~\ref{greensym}, we get a symmetric matrix of dimensions 
\[S(n,d)\le 2\left[\binom{n+d+1}{n+1}-\binom{n+d-1}{n+1}-\binom{n+d-1}{n-1}-1\right].\]
To complete the proof, it is sufficient to use Pascal's formula twice: 
\begin{align*}
\binom{n+d+1}{n+1} &= \binom{n+d}{n+1}+\binom{n+d}{n}\\
    &= \left[\binom{n+d-1}{n+1}+\binom{n+d-1}{n}\right]+\left[\binom{n+d-1}{n}+\binom{n+d-1}{n-1}\right]\\
    &= 2\binom{n+d-1}{n} + \binom{n+d-1}{n+1}+\binom{n+d-1}{n-1}.
\end{align*}

\end{proof}

Note that the bound $F(n,d)$ we obtain with this construction is only better by a linear factor in $n$ than the obvious formula consisting of a sum of 
all the monomials. Indeed, for any $j\le d$, there are at most $\tbinom{n+j-1}{j}$ monomials of degree $j$ which use $(j-1)$ multiplications, and there are at most $(\tbinom{n+d}{d}-1)$ additions. Therefore the size of the formula we get in this way is
\[\sum_{j=1}^d(j-1)\binom{n+j-1}{j}+\binom{n+d}{d}-1=n\binom{n+d}{n+1}=\frac{n(n+d)}{n+1}\binom{n+d-1}{n}.\]
The first equality comes from similar techniques as in the previous proof and the second one is a straightforward computation. This yields a matrix of dimensions 
$\tfrac{n(n+d)}{2(n+1)}S(n,d)$ approximately.

Nevertheless, this is a bound in the worst case, that is for a polynomial $M_{n,d}$ in which all the monomials of degree at most $d$ appear. But in this special case one can change this construction if the aim is to have the polynomial $M_{n,d}$ itself. Indeed, the recurrence given by Equation~\eqref{eq:rec} can be altered 
in the following manner:
\begin{align*}
M_{n,d}&=x_n M_{n,d-1}+M_{n-1,d}\\
 &=x_nM_{n,d-1}+x_{n-1}M_{n-1,d-1}+M_{n-2,d}\\
 &=x_nM_{n,d-1}+\cdots+x_0M_{0,d-1}.
\end{align*}
This gives an inductive construction of a skew circuit to compute $M_{n,d}$. At step $1$, $M_{n,1}$ is built, and it is clear that every $M_{n-k,1}$ is represented by a gate in the circuit. At step $\delta\le d$, suppose that we have a circuit such that every $M_{n-k,\delta-1}$ is represented by a gate. Then one can build a circuit with $(n+1)$ new variable inputs, $(n+1)$ multiplication gates and $n$ addition gates such that every $M_{n-k,\delta}$ is represented by a gate. At each step, the circuit size increases by $(2n+1)$ and $(n+1)$ inputs are added. As the size of the circuit for degree $1$ is $n$ with $(n+1)$ inputs, the circuit for $M_{n,d}$ has size $(2nd-n+d-1)$ and has $(n+1)d$ inputs. This yields a matrix of polynomial dimensions (in $n$ and $d$), 
much smaller than with Quarez's construction.

Let us now compare the bounds of Theorem~\ref{thm:comp} in the worst case with Quarez's. To this end let us consider a polynomial with $n$ variables and of degree $2d$. Then Quarez builds a symmetric matrix of dimensions 
$2\tbinom{n+d}{n}$ whereas our construction yields a matrix of dimensions 
$4\tbinom{n+2d-1}{n}-2$. A bound on the quotient of those quantities can be given using the inequalities (see e.g. \citep{KnuthVol1})
\[\binom{n+d}{n}\le\left(\frac{e(n+d)}{n}\right)^n\text{ and }\binom{n+2d-1}{n}\ge\left(\frac{n+2d-1}{n}\right)^n.\]
So, the quotient is bounded by
\[\left(\frac{e(n+d)}{n}\right)^n\cdot\left(\frac{n}{n+2d-1}\right)^n=e^n\cdot\left(\frac{n+d}{n+2d-1}\right)^n\le e^n.\]
This means that Quarez's construction is exponentially better in the general case even though our construction yields much smaller matrices when the polynomial has a polynomial size formula or weakly-skew circuit. 

We now compare Quarez's results and ours for the special case of the permanent. This is an important example of a polynomial for which no polynomial size circuit is known (even non weakly-skew). Nevertheless, there exist formulas for computing it of much smaller size than the bounds for the general case \citep{Ryser,Gly10}. For instance, Ryser's formula 
to compute the permanent of a matrix $M$ is
\[\per(A) = \sum_{S\subseteq\{1,\dots,n\}} (-1)^{|S|} \prod_{i=1}^n \sum_{j\notin S} M_{ij}.\]
As the sums of variables are not counted, this gives a size-$O(n2^n)$ formula, and hence yields a symmetric matrix of dimensions $O(n2^n)$ 
to represent the permanent. Let us consider the permanent of a $(2n\times 2n)$ matrix. This is a polynomial of degree $2n$ with $4n^2$ variables. Therefore, Quarez's construction yields a matrix of dimensions 
$2\tbinom{4n^2+n}{n}$. This quantity can be bounded as above and therefore we get the following bound (up to a constant factor) for the quotient:
\[\frac{\binom{4n^2+n}{n}}{n2^{2n}}\ge \frac{\left((4n^2+n)/n\right)^n}{n4^n}\ge\frac{4^nn^n}{n4^n}=n^{n-1}.\]
A more careful computation 
\emph{via} Stirling's formula 
shows that this quotient is equal to $O(n^{n-1/2}(4e)^n)$ when $n$ tends to infinity.

\section{Characteristic 2} 
\label{sec:char2}

In characteristic $2$, the constructions of Sections~\ref{sec:formulas} and \ref{sec:ws-circuits} are not valid anymore because of the coefficients $1/2$ they use. Nevertheless, for a polynomial computable by a weakly-skew circuit, it is possible to represent its square as the determinant of a symmetric matrix. On the other hand, representing the polynomial itself seems to be a challenging problem. 
For instance, it is not possible to represent the polynomial $xy+z$ this way (\cite{GMT}), but we don't have for the moment a characterisation of representable polynomials. 
Related to these problems, the $\VNP$-completeness of the partial permanent is also studied. Actually, we give an almost complete answer to an open question of B\"urgisser~\citeyearp{Burgisser} (Problem~3.1) showing that if the partial permanent is complete in finite fields of characteristic $2$, then the (boolean) polynomial hierarchy collapses. For any field of characteristic $2$ (finite or infinite), we show that the $\VNP$-completeness of this family would imply that every $\VNP$ family of polynomials has its square in $\VP_{ws}$ (\emph{i.e.} has polynomial size weakly-skew circuits). 
This also seems unlikely to happen unless $\VP_{ws}=\VNP$. 
We refer to \citep{Burgisser,MP08} for the formal definitions of the complexity classes $\VNP$ and $\VP_{ws}$.

Let $G$ be an edge-weighted graph with vertices $\{v_1,\dots,v_n\}$. Recall that the adjacency matrix $A$ of $G$ is the $(n\times n)$ symmetric matrix defined by $A_{ij}=A_{ji}=w_{ij}$ where $w_{ij}$ is the weight of the edge $v_iv_j$. Suppose now that $G$ is bipartite with two independent sets of vertices $V_r$ and $V_c$ of cardinality $m$ and $n$ respectively. Let $V_r=\{r_1,\ldots,r_m\}$ and $V_c=\{c_1,\ldots,c_n\}$. The \emph{biadjacency matrix} of $G$ (also known as the \emph{bipartite adjacency matrix}) is the $(m\times n)$ matrix $B$ such that $B_{ij}$ is the weight of the edge between $r_i$ and $c_j$. This means that the rows of $B$ are indexed by $V_r$ and its columns by $V_c$. 
For a bipartite graph $G$ of adjacency and biadjacency matrices $A$ and $B$ respectively,
\[A=\begin{bmatrix} 0&B\\B^t&0\end{bmatrix}.\]

Throughout this section, we shall use the usual definition of the weight of a partial matching: it is the product of the weights of the edges it uses.

\subsection{Symmetric determinantal representation of the square of a polynomial} 

\begin{lemma}\label{lemma:cycle-cover}
Let $G$ be an edge-weighted graph and $A$ its adjacency matrix. In characteristic $2$, the determinant of $A$ is the sum of the weights of the cycle covers with cycles of length at most $2$.
\end{lemma}

\begin{proof}
Let us consider $G$ as a symmetric digraph (that is an edge $uv$ is seen as both arcs $(u,v)$ and $(v,u)$). In Lemma~\ref{lemma:perm-det}, the signs of the cycle covers are considered. In characteristic $2$, this is irrelevant. Therefore, the determinant of $A$ is the sum of the weights of the cycle covers of $G$.

Let $C$ be a cycle cover of $G$ containing a (directed) cycle of length at least $3$ denoted by $(v_1,v_2,\dots,v_k,v_1)$. One can change the direction of this cycle (as $G$ is symmetric) and obtain a new cycle cover $C'$ containing the same cycles as $C$, but $(v_k,v_{k-1},\ldots,v_1,v_k)$ instead of $(v_1,v_2,\dots,v_k,v_1)$. Clearly, the weights of $C$ and $C'$ are the same as the graph is symmetric. Therefore, when the determinant of $A$ is computed in characteristic $2$, the contributions of those two cycle covers to the sum cancel out. This shows that the determinant of a matrix in characteristic two is obtained as the sum of the weights of cycle covers with cycles of length $1$ (loops) or $2$.
\end{proof}

\begin{proposition}
Let $p$ be a polynomial over a field of characteristic $2$, represented by a weakly-skew circuit of fat size $m$. Then there exists a symmetric matrix $A$ of dimensions 
$(2m+2)$ such that $p^2=\det(A)$.
\end{proposition}

\begin{proof}
Let $C$ be a weakly-skew circuit representing a polynomial $p$ over a field of characteristic $2$. Let $M$ be the matrix obtained by Malod and Portier's construction~\citeyearp{MP08} such that $\per M=p$. Let $G$ be the digraph represented by $M$, and let $G'$ be the bipartite graph obtained from $G$ by the two following operations: Each vertex $v$ of $G$ is turned into two vertices $v^s$ and $v^t$ in $G'$, and each arc $(u,v)$ is turned into the edge $\{u^s,v^t\}$. A loop on a vertex $u$ is simply represented as the edge $\{u^s,u^t\}$. Let $A$ be the symmetric adjacency matrix of $G'$ (when the vertices are ordered $v_0^s,v_1^s,\ldots,v_m^s,v_0^t,\ldots,v_m^t$). 

It is well-known that cycle covers of $G$ and perfect matchings of $G'$ are in one-to-one correspondence. If there is a cycle cover of $G$, then each vertex $v$ belongs to a cycle, and thus has both a predecessor $v$ and a successor $w$. This means that $u^t$ and $u^s$ are matched to $v^s$ and $w^t$ respectively (if $u$ is covered by a loop, then $u^s$ and $u^t$ are matched). Conversely, suppose that $G'$ has a perfect matching. Let $u^s$ be any vertex. Then it is matched to some $v^t$. In the same way, $v^s$ is matched to some $w^t$. As the set of vertices is finite, as some point we go back to $u^t$. Thus it defines a cycle in $G$, and by doing the same process with other vertices not in this cycle this eventually defines a cycle cover in $G$.

This one-to-one correspondence shows that the determinant of $M$ equals the sum of the weights of the perfect matchings in $G'$. If a perfect matching in $G'$ is considered as a cycle cover with length-$2$ cycles, the weight of the cycle cover is the square of the weight of the perfect matching. Indeed, in the cycle cover, all the arcs of the length-$2$ cycles have to be considered, that is each edge contributes twice to the product. Lemma~\ref{lemma:cycle-cover} and the fact that there is no loop in $G'$ show that 
\[\det(A)=\sum_\mu w(\mu)^2=\Bigl(\sum_\mu w(\mu)\Bigr)^2,\]
where $\mu$ ranges over all perfect matchings of $G'$ and $w(\mu)$ is the weight of the perfect matching $\mu$. The second equality holds as the field has characteristic $2$. 

Finally, it is shown in \citep{MP08} that $p=\det(M)$, and we showed that $\det(M)=\sum_\mu w(\mu)$ and $\det(A)=\bigl(\sum_\mu w(\mu)\bigr)^2$. Therefore, $\det(A)=\det(M)^2=p^2$.
\end{proof}

This proposition raises the following question: Let $f$ be a family of polynomials such that $f^2\in\VP_{ws}$. Does $f$ belong to $\VP_{ws}$? This question is discussed with more details in the next section.

\subsection{Is the partial permanent complete in characteristic $2$?} 
\label{sec:partial-perm}

\begin{definition}
Let $X=(X_{ij})$ be an $(n\times n)$ matrix. The partial permanent of $X$, as defined by Bürgisser~\citeyearp{Burgisser}, is 
\[\per^* (X) = \sum_\pi \prod_{i\in\domdef(\pi)}X_{i\pi(i)},\]
where the sum ranges over the injective partial maps from $[n]=\{1,\dots,n\}$ to $[n]$ and $\domdef(\pi)$ is the domain of the partial map $\pi$ 
(recall that a partial map is a map from a subset of $[n]$ to $[n]$).

The family $(\PER^*_n)$ is the family of polynomials such that $\PER^*_n$ is the partial permanent of the $(n\times n)$ matrix whose coefficients are the indeterminates $X_{ij}$.
\end{definition}

\begin{lemma}
Let $G$ be the complete bipartite graph with two independent sets of vertices $V_r$ and $V_c$ such that the edge between $r_i$ and $c_j$ is labelled by $B_{ij}$ (the matrix $B$ is the biadjacency matrix of $G$). Then the partial permanent of $B$ is equal to the sum of the weights of the partial matchings of $G$.
\end{lemma}

A partial matching in a graph $G$ is a set of pairs of vertices connected by an edge such that no vertex appears in more than a pair. Equivalently, a partial matching can be seen as a set of disjoint edges. The weight of a partial matching is the product of the weights of its edges.

The proof of the lemma is quite straightforward as a partial injective map $\pi$ from $[n]$ to $[n]$ exactly defines a partial matching in $G$ such that for $i\in\domdef(\pi)$, $r_i$ is matched with $c_{\pi(i)}$.

\begin{lemma} \label{partial2}
Let $G$ be the complete bipartite graph with two independent sets of vertices $V_r$ and $V_c$ such that the edge between $r_i$ and $c_j$ is labelled by $B_{ij}$ (the matrix $B$ is the biadjacency matrix of $G$). Let $A$ be its adjacency matrix. Then in characteristic $2$, 
\[\det(A+I_{2n}) = (\per^*(B))^2,\]
where $I_{2n}$ is the identity matrix of dimensions 
$2n$.
\end{lemma}

\begin{proof}
By Lemma~\ref{lemma:cycle-cover}, to compute a determinant in characteristic $2$, one can focus only on cycles of length at most $2$. A cycle cover with such cycles actually is a partial matching when the graph is symmetric (length-$2$ cycles define the pairs of vertices, and length-$1$ cycles are isolated vertices). Considering $G$ as a symmetric digraph, the weight of a cycle cover is equal to the product of the weights of its loops and the square of the weights of the edges it uses (a length-$2$ cycle corresponds to an edge).

Consider the graph $G'$ obtained from $G$ by adding weight-$1$ loops on all its vertices. In other words, $G'$ is the graph whose adjacency matrix is $A+I_{2n}$. By the previous remark, and by the fact that the loops have weight $1$, the determinant of $A+I_{2n}$ is 
\[\det(A+I_{2n})=\sum_\mu w(\mu)^2=\Bigl(\sum_\mu w(\mu)\Bigr)^2\]
where $\mu$ ranges over the partial matchings of $G'$ and $w(\mu)$ is the weight of the partial matching $\mu$. The second equality is true as the characteristic of the field is $2$.

Recall now that $G$ is bipartite. Of course, the partial matchings of $G$ and $G'$ are the same. So 
\[\per^*(B)=\sum_\mu w(\mu),\]
where $\mu$ ranges over the partial matchings of $G$. This proves the lemma.
\end{proof}
An alternative proof of this lemma was suggested by an anonymous referee.
In any field, we have the polynomial identity 
$\per(A+tI_{2n})=\sum_{k=0}^{2n} c_k t^{2n-k}$, where $c_k$ is the sum 
of the permanents of all central minors of $A$ of size $k$.
In particular, we have $\det(A+I_{2n})=\sum_{k=0}^{2n} c_k$ 
in characteristic~2. A nonzero permanent of a central minor of $A$ is 
of the form 
\[\per \begin{bmatrix}
0 & M\\
M^T & 0
\end{bmatrix}
= \per(M)^2,\]
where $M$ is a square submatrix of $B$.
Hence $\det(A+I_{2n}) = \sum_{M \subseteq B} \per(M)^2$, where $M$ ranges
over all square submatrices of $B$. Since we are in characteristic~2,
this last sum is equal to $(\sum_{M \subseteq B} \per(M))^2$.
But we have 
$\per^*(B)=\sum_{M \subseteq B} \per(M)$  
by definition of the partial permanent, 
and the conclusion of the lemma follows.

Lemma~\ref{partial2} shows in particular that to compute the parity of the number of partial matchings in a bipartite graph, it is sufficient to compute a determinant (this is the case where $G$ is not edge-weighted). Therefore, this problem is solvable in polynomial time. This was already mentioned by Valiant~\citeyearp{Val05} but without any proof or reference.

\begin{theorem}
In characteristic $2$, the family $((\PER^*)^2_n)$ is in $\VP_{ws}$.
\end{theorem}

\begin{proof}\sloppy
The previous lemma shows that the polynomial $(\PER^*)^2_n$ is a $p$-projection of $\DETfam_{2n}$ in characteristic $2$. Thus, $((\PER^*)^2_n)$ is in $\VP_{ws}$.
\end{proof}

Suppose that $(\PER^*_n)$ is $\VNP$-complete. Then every $\VNP$ family $(f_n)$ is a $p$-projection of $(\PER^*_n)$, and thus $(f_n^2)$ is a $p$-projection of $((\PER^*)^2_n)$. Let $\VNP^2=\{(f_n^2) : (f_n)\in\VNP\}$ be the class of \emph{squares of $\VNP$ families}. This implies the following corollary of the theorem:

\begin{corollary}\label{cor:square}
In any field of characteristic $2$, if $(\PER^*_n)$ is $\VNP$-complete, then $\VNP^2\subseteq \VP_{ws}$.
\end{corollary}

This situation is unlikely to happen. In particular, it would be interesting to investigate whether this inclusion implies that $\VP_{ws}=\VNP$ in characteristic $2$. Let us now give another consequence of $(\PER^*_n)$ being $\VNP$-complete. This only holds for finite fields of characteristic $2$ but may give a stronger evidence that $(\PER^*_n)$ is unlikely to be $\VNP$-complete.

\begin{theorem}
If the partial permanent family is $\VNP$-complete in a finite field of characteristic $2$, then $\parity\Ppoly=\NC^2/\poly$, and the polynomial hierarchy collapses to the second level.
\end{theorem}

The proof of this theorem uses the \emph{boolean parts} of Valiant's complexity classes defined in \citep{Burgisser}. In the context of finite fields of characteristic $2$, the boolean part of a family $(f_n)$ of polynomials with coefficients in the ground field $\FF_2$ is the function $bp_f:\{0,1\}^*\to\{0,1\}$ such that for $x\in\{0,1\}^n$, $bp_f(x)=f_n(x) \pmod 2$. 
The boolean part $\BoolP(C)$ of a Valiant's class $C$ is the set of boolean parts of all $f\in C$.

\begin{proof}

Let $(f_n)$ be a $\VNP$ family and $(\varphi_n)$ its boolean part. As $\varphi_n(x)\in\{0,1\}$ for all $x\in\{0,1\}^n$, $(\varphi_n)$ is the boolean part of $(f_n^2)$ too. This shows that $\BoolP(\VNP)\subseteq\BoolP(\VNP^2)$. 
By Corollary~\ref{cor:square}, $\VNP^2\subseteq\VP_{ws}\subseteq\VP$.  Thus, $\BoolP(\VNP)\subseteq\BoolP(\VNP^2)\subseteq\BoolP(\VP)$ and as $\VP\subseteq\VNP$
\[\BoolP(\VP)=\BoolP(\VNP).\]

B\"urgisser~\citeyearp{Burgisser} shows that in a finite field of characteristic $2$, $\parity\Ppoly=\BoolP(\VNP)$, and $\BoolP(\VP)\subseteq\NC^2/\poly$. Hence, $\parity\Ppoly\subseteq\NC^2/\poly$. Moreover, $\NC^2/\poly\subseteq\Ppoly\subseteq\parity\Ppoly$ hence we conclude that
\[\parity\Ppoly=\NC^2/\poly.\]

The collapse of the polynomial hierarchy follows from a non uniform version of the Valiant-Vazirani Theorem~\citeyearp{ValVaz86}: {Theorem~4.10} in \citep{Burgisser} states that $\NP/\poly\subseteq\parity\Ppoly$. Therefore, 
\[\NC^2/\poly\subseteq\NP/\poly\subseteq\parity\Ppoly=\NC^2/\poly.\]
In particular, $\Ppoly=\NP/\poly$ and Karp and Lipton \citeyearp{KL82} showed that this implies the collapse of the polynomial hierarchy to the second level.
\end{proof}

Since the submission of this paper, B\"urgisser's open problem has been completely settled. Guillaume Malod~\citeyearp{Mal11} has proved, using clow sequences \emph{\`a la} Mahajan and Vinay~\citeyearp{MV97}, that $\PER^*\in\VP_{ws}$. Stefan Mengel subsequently noticed that the result can be derived from a result of Valiant on Pfaffian Sums~\citep{Val02}, see also \citep{GLV11}.

\section{Conclusion} 

Figure~\ref{fig:graphs} shows 
the graphs obtained from the weakly-skew circuit and the formula of Fig.~\ref{fig:circuit-weak-form}\subref{fig:ex-weak} and \subref{fig:ex-form} 
for a field of characteristic different from $2$, and Table~\ref{tab:summary} recalls all the constructions used in this paper.

    \begin{figure}[htbp]
    \begin{center}
    \input{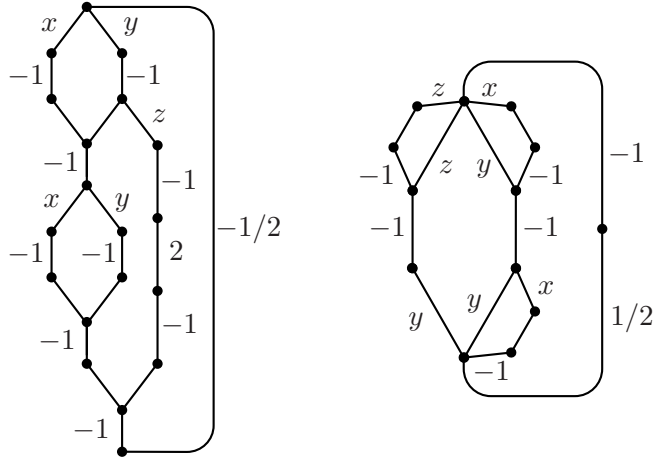}
    \caption{Graphs obtained from the weakly-skew circuit and the formula  given in Fig.~\ref{fig:circuit-weak-form}\subref{fig:ex-weak} and \subref{fig:ex-form}.}
    \label{fig:graphs}
    \end{center}
    \end{figure}

Table~\ref{tab:comp} compares the results obtained, in this paper and in previous ones. The bounds are given for a formula of green size $e$ and for a weakly-skew circuit of green size $e$ with $i$ input gates labelled by a variable.

\begin{table}[htbp]
\begin{minipage}[t]{\textwidth}
\renewcommand{\footnoterule}{}
\begin{center}
\begin{tabular}[c]{|l||c|c|}
\hline
 & Non-symmetric & Symmetric \\
 & matrix        & matrix    \\\hline
Formula & $e+1$ & $2e+1$\footnote{The bound is achieved if and only if the entries can be complex numbers. Else, the bound is $2e+2$.}\\\hline
Weakly-skew circuit & $(e+i)+1$ & $2(e+i)+1$\\\hline
\end{tabular}
\end{center}
\end{minipage}
\caption{Bounds for determinantal representations of formulas and weakly-skew circuits.
The bounds for symmetric representations are new, and the bound for a non-symmetric representation 
of a weakly-skew circuit is a slight improvement of known bounds.}
\label{tab:comp}
\end{table}

\begin{landscape}
\newcommand\tab[2]{\parbox[c]{#2}{\input{tab-#1.svg_tex}} } 
\begin{table}[p]\centering
\begin{tabular}{|p{15ex}||c||c|c||c|c|} \hline
& Valiant's construction & Formulas & Formulas       & Weakly skew & Weakly skew circuits  \\
& with constants         &          & with constants & circuits    & with constants \\
&(Section~\ref{sec:valiant})&(Section~\ref{sec:sym-form})&(Section~\ref{sec:sym-form})&(Section~\ref{sec:fat-size})&(Section~\ref{sec:green-size})\\
    \hline\hline
\centering Input gate & \parbox[c]{16pt}{\input{tab-val-base.svg_tex}}  & \parbox[c]{16pt}{\input{tab-form-base.svg_tex}}  & \parbox[c]{16pt}{\input{tab-form-base.svg_tex}}  & \parbox[c]{89.6pt}{\input{tab-weak-base.svg_tex}}  & \parbox[c]{89.6pt}{\input{tab-weak-base.svg_tex}}  \\
           & constant $1$ & no constant & constant $1$ & no constant & constants $1$ \\\hline
\centering Addition gate & \parbox[c]{68pt}{\input{tab-val-sum.svg_tex}} &\parbox[c]{49.6pt}{\input{tab-form-sum.svg_tex}} &\parbox[c]{61.6pt}{\input{tab-form-sum-cst.svg_tex}} &\parbox[c]{95.2pt}{\input{tab-weak-sum.svg_tex}} &\parbox[c]{119.2pt}{\input{tab-weak-sum-csts.svg_tex}} \\
              & constant $c_1$ & no constant & constant $c_1$ & no constant & constants $1$ \\\hline
\centering Multiplication gate & \parbox[c]{22.4pt}{\input{tab-val-prod.svg_tex}}  & \parbox[c]{27.2pt}{\input{tab-form-prod.svg_tex}}  & \parbox[c]{27.2pt}{\input{tab-form-prod.svg_tex}}  & \parbox[c]{44pt}{\input{tab-weak-prod.svg_tex}}  & \parbox[c]{44pt}{\input{tab-weak-prod.svg_tex}} \\
                    & constant $c_1c_2$ & no constant & constant $c_1c_2$ & no constant & constant $c_1c_2c_\beta c_\gamma$ \\\hline
\end{tabular}
\caption{Summary of the constructions} 
\label{tab:summary}
\end{table}
\end{landscape}

The $(e+1)$ bound for the representation of a formula by a (non-symmetric) matrix determinant was given in \citep{LR06} by a method purely based on matrices. We show in Section~\ref{sec:valiant} that this bound can also be obtained directly from Valiant's original proof when we remove the little flaw it contains.
The $(e+i+1)$ bound for the representation of a polynomial computed by a weakly-skew circuit can be obtained from the $(m+1)$ bound (where $m$ is the fat size of the circuit) obtained in \citep{MP08} if we use our minimization lemma (Lemma~\ref{lemma:minimization}) as well as a similar trick as in the proof of Theorem~\ref{thm:weak+csts}. Both bounds for the symmetric cases are given in this paper.

A formula is a special case of weakly-skew circuit. If our construction for weakly-skew circuits is applied to a formula, this yields a matrix that can be as large as twice the size of the matrix obtained with the specific constructions for the formulas. In the converse way, one could turn a weakly-skew circuit into a formula and then apply the construction for the formula. Yet, turning a weakly-skew circuit into a formula of polynomial size is not known to be possible. In fact, this would give a polynomial size formula for the determinant, and hence a parallel time 
upper bound of $O(\log n)$. So far, the best upper bound is Csansky's famous $O(\log^2 n)$ upper bound~\citeyearp{Csa76}.

All of these results are valid for any field of characteristic different from $2$. We showed that there are some important differences 
for the complexity of polynomials over fields of characteristic $2$.
The question of characterizing which polynomials can be represented as determinants of symmetric matrices is quite intriguing and remains open.

\bibliographystyle{myplainnat}
\bibliography{biblio}

\begin{thebibliography}{31}
\expandafter\ifx\csname natexlab\endcsname\relax\def\natexlab#1{#1}\fi
\expandafter\ifx\csname url\endcsname\relax
  \def\url#1{{\tt #1}}\fi

\bibitem[Beimel and G{\'a}l(1999)]{BG99}
Beimel, Amos and G{\'a}l, Anna.
\newblock On arithmetic branching programs.
\newblock {\em J. Comput. System Sci.}, 59\penalty0 (2):\penalty0 195--220,
  1999.

\bibitem[Berkowitz(1984)]{Ber84}
Berkowitz, Stuart~J.
\newblock On computing the determinant in small parallel time using a small
  number of processors.
\newblock {\em Inform. Process. Lett.}, 18:\penalty0 147--150, 1984.

\bibitem[{Br{\"a}nd{\'e}n}(2011)]{Branden2010}
{Br{\"a}nd{\'e}n}, P.
\newblock {Obstructions to determinantal representability}.
\newblock {\em Adv. Math.}, 226\penalty0 (2):\penalty0 1202 -- 1212, 2011.
\newblock \url{http://adsabs.harvard.edu/abs/2010arXiv1004.1382B}.

\bibitem[B\"urgisser(2000)]{Burgisser}
B\"urgisser, Peter.
\newblock {\em Completeness and Reduction in Algebraic Complexity Theory}.
\newblock Algorithms and Computation in Mathematics. Springer, 2000.
\newblock ISBN 9783540667520.

\bibitem[B\"urgisser et~al.(1997)B\"urgisser, Clausen, and
  Shokrollahi]{BurgisserClausenShokrollahi}
B\"urgisser, Peter, Clausen, Michael, and Shokrollahi, Mohammad~A.
\newblock {\em Algebraic Complexity Theory}, volume 315 of {\em Grundlehren
  Math. Wiss.}
\newblock Springer, 1997.
\newblock ISBN 3540605827.

\bibitem[Csanky(1976)]{Csa76}
Csanky, L.
\newblock Fast parallel matrix inversion algorithms.
\newblock {\em SIAM J. Comput.}, 5\penalty0 (4):\penalty0 618--623, 1976.

\bibitem[Glynn(2010)]{Gly10}
Glynn, D.G.
\newblock {The permanent of a square matrix}.
\newblock {\em European J. Combin.}, 31\penalty0 (7):\penalty0 1887--1891,
  2010.
\newblock ISSN 0195-6698.

\bibitem[Grenet et~al.(2011{\natexlab{a}})Grenet, Kaltofen, Koiran, and
  Portier]{GKKP11}
Grenet, Bruno, Kaltofen, Erich~L., Koiran, Pascal, and Portier, Natacha.
\newblock {S}ymmetric {D}eterminantal {R}epresentation of {W}eakly-{S}kew
  {C}ircuits.
\newblock In Schwentick, Thomas and D{\"u}rr, Christoph, editors, {\em Proc.
  28th STACS}, number~9 in LIPIcs, pages 543--554. Schloss
  Dagstuhl--Leibniz-Zentrum f{\"u}r Informatik, 2011{\natexlab{a}}.

\bibitem[Grenet et~al.(2011{\natexlab{b}})Grenet, Monteil, and Thomass\'e]{GMT}
Grenet, Bruno, Monteil, Thierry, and Thomass\'e, Stephan.
\newblock Symmetric determinantal representations in characteristic 2.
\newblock In preparation, 2011{\natexlab{b}}.

\bibitem[Guo et~al.(2011)Guo, Lu, and Valiant]{GLV11}
Guo, Heng, Lu, Pinyan, and Valiant, Leslie~G.
\newblock {The Complexity of Symmetric Boolean Parity Holant Problems (Extended
  Abstract)}.
\newblock In {\em Proc. 38th ICALP}, 2011.
\newblock to appear.

\bibitem[Helton and Vinnikov(2006)]{HV06}
Helton, J.W. and Vinnikov, V.
\newblock {Linear matrix inequality representation of sets}.
\newblock {\em Comm. Pure Appl. Math.}, 60\penalty0 (5):\penalty0 654--674,
  2006.
\newblock \url{http://arxiv.org/pdf/math.OC/0306180}.

\bibitem[Helton et~al.(2006)Helton, McCullough, and Vinnikov]{HMcCV06}
Helton, J.~William, McCullough, Scott~A., and Vinnikov, Victor.
\newblock Noncommutative convexity arises from linear matrix inequalities.
\newblock {\em J. Funct. Anal.}, 240\penalty0 (1):\penalty0 105--191, November
  2006.
\newblock \url{http://math.ucsd.edu/~helton/osiris/NONCOMMINEQ/convRat.ps}.

\bibitem[Kaltofen(1992)]{Ka92:issac}
Kaltofen, E.
\newblock On computing determinants of matrices without divisions.
\newblock In Wang, P.~S., editor, {\em Proc ISSAC'92}, pages 342--349, New
  York, N. Y., 1992. ACM Press.
\newblock
  \EKhref{http://www.math.ncsu.edu/~kaltofen/bibliography/92/Ka92\_issac.pdf}
  {EKbib/92/Ka92\_issac.pdf}.

\bibitem[Kaltofen and Koiran(2008)]{KaKoi08}
Kaltofen, Erich and Koiran, Pascal.
\newblock Expressing a fraction of two determinants as a determinant.
\newblock In Jeffrey, David, editor, {\em {P}roc {ISSAC}'08}, pages 141--146,
  New York, N. Y., 2008. ACM Press.
\newblock ISBN 978-1-59593-904-3.
\newblock
  \EKhref{http://www.math.ncsu.edu/~kaltofen/bibliography/08/KaKoi08.pdf}
  {EKbib/08/KaKoi08.pdf}.

\bibitem[Kaltofen and Villard(2004)]{KaVi04:2697263}
Kaltofen, Erich and Villard, Gilles.
\newblock On the complexity of computing determinants.
\newblock {\em Comput. Complexity}, 13\penalty0 (3-4):\penalty0 91--130, 2004.
\newblock
  \EKhref{http://www.math.ncsu.edu/~kaltofen/bibliography/04/KaVi04\_2697263.p%
df} {EKbib/04/KaVi04\_2697263.pdf}.

\bibitem[Karp and Lipton(1982)]{KL82}
Karp, R.M. and Lipton, R.J.
\newblock {Turing machines that take advice}.
\newblock {\em Enseign. Math.}, 28:\penalty0 191--209, 1982.

\bibitem[Knuth(1997)]{KnuthVol1}
Knuth, Donald~E.
\newblock {\em The Art of Computer Programming, Volume 1: Fundamental
  Algorithms (3rd Edition)}.
\newblock Addison-Wesley Professional, 3rd edition, 1997.
\newblock ISBN 9780201896831.

\bibitem[Lewis et~al.(2005)Lewis, Parrilo, and Ramana]{LPR05}
Lewis, A.S., Parrilo, P.A., and Ramana, M.V.
\newblock {The Lax conjecture is true}.
\newblock {\em Proc. Amer. Math. Soc.}, 133\penalty0 (9):\penalty0 2495--2500,
  2005.
\newblock \url{http://arxiv.org/pdf/math.OC/0304104}.

\bibitem[Liu and Regan(2006)]{LR06}
Liu, H. and Regan, K.W.
\newblock {Improved construction for universality of determinant and
  permanent}.
\newblock {\em Inform. Process. Lett.}, 100\penalty0 (6):\penalty0 233--237,
  2006.

\bibitem[Mahajan and Vinay(1997)]{MV97}
Mahajan, M. and Vinay, V.
\newblock {Determinant: Combinatorics, algorithms, and complexity}.
\newblock {\em Chic. J. Theoret. Comput. Sci.}, 5\penalty0 (1997):\penalty0
  730--738, 1997.

\bibitem[Malod(2011)]{Mal11}
Malod, G.
\newblock Computing the partial permanent in characteristic $2$.
\newblock Unpublished manuscript, 2011.

\bibitem[Malod and Portier(2008)]{MP08}
Malod, G. and Portier, N.
\newblock {Characterizing Valiant's algebraic complexity classes}.
\newblock {\em J. Complexity}, 24\penalty0 (1):\penalty0 16--38, 2008.
\newblock Presented at MFCS'06.

\bibitem[Nisan(1991)]{N91}
Nisan, Noam.
\newblock Lower bounds for non-commutative computation.
\newblock In {\em Proc. 23rd STOC}, pages 410--418. ACM, 1991.

\bibitem[Quarez(2008)]{Quarez08}
Quarez, Ronan.
\newblock Symmetric determinantal representation of polynomials.
\newblock \url{http://hal.archives-ouvertes.fr/hal-00275615/en/}, April 2008.

\bibitem[Ryser(1963)]{Ryser}
Ryser, Herbert~J.
\newblock {\em Combinatorial Mathematics}, volume~14 of {\em Carus Math.
  Monogr.}
\newblock Mathematical Association of America, Washington, 1963.
\newblock ISBN 0883850141.

\bibitem[Toda(1992)]{Tod92}
Toda, S.
\newblock {Classes of arithmetic circuits capturing the complexity of computing
  the determinant}.
\newblock {\em IEICE T. Inf. Syst.}, 75\penalty0 (1):\penalty0 116--124, 1992.

\bibitem[Valiant(2002)]{Val02}
Valiant, L.G.
\newblock {Quantum circuits that can be simulated classically in polynomial
  time}.
\newblock {\em SIAM J. Comput.}, 31:\penalty0 1229, 2002.

\bibitem[Valiant(2005)]{Val05}
Valiant, L.G.
\newblock {Completeness for parity problems}.
\newblock {\em Computing and Combinatorics}, pages 1--8, 2005.

\bibitem[Valiant and Vazirani(1986)]{ValVaz86}
Valiant, L.G. and Vazirani, V.V.
\newblock {$\NP$ is as easy as detecting unique solutions}.
\newblock {\em Theoret. Comput. Sci.}, 47:\penalty0 85--93, 1986.

\bibitem[Valiant(1979)]{Val79}
Valiant, L.~G.
\newblock Completeness classes in algebra.
\newblock In {\em Proc. 11th STOC}, pages 249--261, New York, N.Y., 1979. ACM.

\bibitem[von~zur Gathen(1987)]{vzG87}
von~zur Gathen, J.
\newblock {Feasible arithmetic computations: Valiant's hypothesis}.
\newblock {\em J. Symbolic Comput.}, 4\penalty0 (2):\penalty0 137--172, 1987.

\end{thebibliography}

\end{document}